\documentclass[11pt,a4paper]{report}

\usepackage[utf8]{inputenc}
\usepackage{a4wide}
\usepackage{mathpazo}
\usepackage{amsmath,amsthm}
\usepackage{graphicx}
\usepackage{color}
\usepackage{soul}
\usepackage{cite}
\usepackage{url}

\newcommand{\R}{\mathbb{R}}
\newcommand{\I}{\mathbb{1}}

\newcommand{\N}{\mathcal{N}}
\newcommand{\Hs}{\mathcal{H}_Z}
\renewcommand{\L}{\mathcal{L}}

\newcommand{\A}{\mathbf{A}}
\newcommand{\B}{\mathbf{B}}
\newcommand{\C}{\mathbf{C}}
\newcommand{\D}{\mathbf{D}}
\newcommand{\X}{\mathbf{X}}
\newcommand{\Y}{\mathbf{Y}}
\newcommand{\Z}{\mathbf{Z}}
\newcommand{\K}{\mathbf{K}}
\newcommand{\U}{\mathbf{U}}
\newcommand{\V}{\mathbf{V}}
\newcommand{\W}{\mathbf{W}}
\newcommand{\Pb}{\mathbf{P}_Z}
\newcommand{\M}{\mathbf{M}}
\renewcommand{\P}{\mathbf{P}}
\newcommand{\Q}{\mathbf{Q}}
\newcommand{\Rm}{\mathbf{R}}
\newcommand{\T}{\mathbf{T}}

\newcommand{\ebf}{\boldsymbol\epsilon}
\newcommand{\Gbf}{\boldsymbol\Gamma}

\DeclareMathOperator*{\argmin}{argmin}

\DeclareMathOperator{\diag}{diag}
\DeclareMathOperator{\cov}{Cov}

\DeclareMathOperator{\tr}{tr}
\DeclareMathOperator{\rnk}{rank}

\newtheorem{thm}{Theorem}
\newtheorem{prop}{Proposition}
\newtheorem{lem}{Lemma}

\newcommand{\lvreml}{\textsc{lvreml}}
\newcommand{\panama}{\textsc{panama}}
\newcommand{\limix}{\textsc{limix}}
\newcommand{\peer}{\textsc{peer}}
\newcommand{\lvurl}{\url{https://github.com/michoel-lab/lvreml}}
\newcommand{\panurl}{\url{https://github.com/limix/limix-legacy}}

\graphicspath{{./}{./fig/}{./notebooks/figures/}{./notebooks/figures/LVREML/}{./notebooks/figures/PANAMA/}}

\renewcommand\thesection{\arabic{section}}

\title{Restricted maximum-likelihood method for learning latent variance components in gene expression data with known and unknown confounders}

\author{Muhammad Ammar Malik and Tom Michoel$^*$}

\begin{document}


\begin{titlepage}
  \chapter*{\centering Restricted maximum-likelihood method for learning latent variance components in gene expression data with known and unknown confounders}

\begin{center}
  {\Large Muhammad Ammar Malik and Tom Michoel$^*$}
\end{center}

\medskip

Computational Biology Unit, Department of Informatics, University of Bergen, PO Box 7803, 5020 Bergen, Norway

$^*$ Corresponding author, email: \texttt{tom.michoel@uib.no}

\section*{\centering Abstract}

Random effect models are popular statistical models for detecting and correcting spurious sample correlations due to hidden confounders in genome-wide gene expression data. In applications where some confounding factors are known, estimating simultaneously the contribution of known and latent variance components in random effect models is a challenge that has so far relied on numerical gradient-based optimizers to maximize the likelihood function. This is unsatisfactory because the resulting solution is poorly characterized and the efficiency of the method may be suboptimal. Here we prove analytically that maximum-likelihood latent variables can always be chosen orthogonal to the known confounding factors, in other words, that maximum-likelihood latent variables explain sample covariances not already explained by known factors. Based on this result we propose a restricted maximum-likelihood method which estimates the latent variables by maximizing the likelihood on the restricted subspace orthogonal to the known confounding factors, and show that this reduces to probabilistic PCA on that subspace. The method then estimates the variance-covariance parameters by maximizing the remaining terms in the likelihood function given the latent variables, using a newly derived analytic solution for this problem. Compared to gradient-based optimizers, our method attains greater or equal likelihood values, can be computed using standard matrix operations, results in latent factors that don't overlap with any known factors, and has a runtime reduced by several orders of magnitude. Hence the restricted maximum-likelihood method  facilitates the application of random effect modelling strategies for learning latent variance components to much larger gene expression datasets than possible with current methods.
	
\end{titlepage}

\section{Introduction}
\label{sec:introduction}

Following the success of genome-wide association studies (GWAS) in mapping the genetic architecture of complex traits and diseases in human and model organisms \cite{mackay2009genetics,hindorff2009potential,manolio2013bringing}, there is now a great interest in complementing these studies with molecular data to understand how genetic variation affects epigenetic and gene expression states \cite{albert2015role,franzen2016,gtex2017genetic}. In GWAS, it is well-known that population structure or cryptic relatedness among individuals may lead to confouding that can alter significantly the outcome of the study \cite{astle2009population}. When dealing with molecular data, this is further exacerbated by the often unknown technical or environmental influences on the data generating process. This problem is not confined to population-based studies -- in single-cell analyses of gene expression, hidden subpopulations of cells and an even greater technical variability cause significant expression heterogeneity that needs to be accounted for \cite{buettner2015computational}.

In GWAS, linear mixed models have been hugely successful in dealing with confounding due to population structure \cite{yu2006unified, astle2009population, kang2010variance, lippert2011fast, zhou2012genome}. In these models, it is assumed that an individual's trait value is a linear function of fixed and random effects, where the random effects are normally distributed with a covariance matrix determined by the genetic similarities between individuals, hence accounting for confounding in the trait data. Random effect models have also become popular in the correction for hidden confounders in gene expression data \cite{kang2008accurate, listgarten2010correction, fusi2012joint}, generally outperforming approaches based on principal component analysis (PCA), the singular value decomposition or other hidden factor models \cite{leek2007capturing,stegle2010bayesian, stegle2012using}. In this context, estimating the latent factors and the sample-to-sample correlations they induce on the observed high-dimensional data is the critical problem to solve.

If it is assumed that the observed correlations between samples are entirely due to latent factors, it can be shown that the resulting random effect model is equivalent to probabilistic PCA, which can be solved analytically in terms of the dominant eigenvectors of the sample covariance matrix \cite{tipping1999probabilistic,lawrence2005probabilistic}. However, in most applications, some confounding factors are known in advance (e.g.\ batch effects, genetic factors in population-based studies, or cell-cycle stage in single-cell studies), and the challenge is to estimate simultaneously the contribution of the known as well as the latent factors. This has so far relied on the use of numerical gradient-based quasi-Newton optimizers to maximize the likelihood function \cite{fusi2012joint,buettner2015computational}. This is unsatisfactory because the resulting solution is poorly characterized,  the relation between the known and latent factors is obscured, and due to the high-dimensionality of the problem, ``limited memory'' optimizers have to be employed whose theoretical convergence guarantees are somewhat weak \cite{liu1989limited,lin2017generic}.

Intuitively, latent variables should explain sample covariances not already explained by known confounding factors. Here we demonstrate analytically that this intuition is correct: latent variables can always be chosen orthogonal to the known factors without reducing the likelihood or variance explained by the model. Based on this result we propose a method that is conceptually analogous to estimating fixed and random effects in linear mixed models using the restricted maximum-likelihood (REML) method, where the variance parameters of the random effects are estimated on the restricted subspace orthogonal to the maximum-likelihood estimates of the fixed effects \cite{gumedze2011parameter}. Our method, called \lvreml, similarly estimates the latent variables by maximizing the likelihood on the restricted subspace orthogonal to the known factors, and we show that this reduces to probabilistic PCA on that subspace. It then estimates the variance-covariance parameters by maximizing the remaining terms in the likelihood function given the latent variables, using a newly derived analytic solution for this problem. Similarly to the REML method for conventional linear mixed models, the \lvreml\ solution is not guaranteed to maximize the total likelihood function. However we prove analytically that for any given number $p$ of latent variables, the \lvreml\ solution attains minimal unexplained variance among all possible choices of $p$ latent variables, arguably a more intuitive and easier to understand criterion.

The inference of latent variables that explain observed sample covariances in gene expression data is usually pursued for two reasons. First, the latent variables, together with the known confounders, are used to construct a sample-to-sample covariance matrix that is used for downstream estimation of variance parameters for individual genes and improved identification of trans-eQTL associations \cite{fusi2012joint,stegle2012using}. Second, the latent variables are used directly as ``endophenotypes'' that are given a biological interpretation and whose genetic architecture is of stand-alone interest \cite{parts2011joint,stegle2012using}. This study contributes to both objectives. First, we show that the covariance matrix inferred by \lvreml\ is \emph{identical} to the one inferred by gradient-based optimizers, while computational runtime is reduced by orders of magnitude (e.g.\ a $10^4$ speed-up on gene expression data from 600 samples). Secondly, latent variables inferred by \lvreml\ by design do not overlap with already known covariates and thus represent \emph{new} aggregate expression phenotypes of potential interest. In contrast, we show that existing methods infer latent variables that overlap significantly with the known covariates (cosine similarities of upto 30\%) and thus represent partially redundant expression phenotypes.


\section{Results}
\label{sec:results}

\subsection{Restricted maximum-likelihood solution for a random effect model with known and latent variance components}
\label{sec:exact-solut-line}

Our model to infer latent variance components in a gene expression data matrix is the same model that was popularized in the  \panama  \cite{fusi2012joint} and scLVM \cite{buettner2015computational} softwares, where a linear relationship is assumed between expression levels and the known and latent factors, with random noise added (SI Section~\ref{sec:model}). In matrix notation, the model can be written as
\begin{equation}\label{eq:38}
  \Y = \Z \V + \X \W + \ebf,
\end{equation}
where $\Y\in\R^{n\times m}$ is a matrix of gene expression data for $m$ genes in $n$ samples, and $\Z\in\R^{n\times d}$ and $\X\in\R^{n\times p}$ are matrices of values for $d$ known and $p$ latent confounders in the same $n$ samples. The columns $v_i$ and $w_i$ of the random matrices $\V\in\R^{d\times m}$ and $\W\in\R^{p\times m}$  are the effects of the known and latent confounders, respectively, on the expression level of gene $i$, and are assumed to be jointly normally distributed:
\begin{equation*}
  p\left(
  \begin{bmatrix}
    v_i\\ w_i
  \end{bmatrix}
  \right) = \N\left(0,
    \begin{bmatrix}
      \B & \D\\
      \D^T & \A
    \end{bmatrix}\right) 
\end{equation*}
where $\B\in\R^{d\times d}$, $\A\in\R^{p\times p}$ and $\D\in\R^{d\times p}$ are the covariances of the known-known, latent-latent and known-latent confouder effects, respectively. Lastly, $\ebf\in\R^{n\times m}$ is a matrix of independent samples of a Gaussian distribution with mean zero and variance $\sigma^2$, independent of the confounding effects.

Previously, this model was considered with independent random effects ($\B$ and $\A$ diagonal and $\D=0$) \cite{fusi2012joint,buettner2015computational}. As presented here, the model is more general and accounts for possible lack of independence between the effects of the known covariates. Furthermore, allowing the effects of the known and latent factors to be dependent ($\D\neq 0$) is precisely what will allow the latent variables to be orthogonal to the known confounders (SI Section \ref{sec:solution-full-model}). An equivalent model with $\D=0$ can be considered, but requires non-orthogonal latent variables to explain part of the sample covariance matrix, resulting in a mathematically less tractable framework. Finally, it remains the case that we can always choose $\A$ to be diagonal, because the latent factors have an inherent rotational symmetry that allows any non-diagonal model to be converted to an equivalent diagonal model (SI Section~\ref{sec:solut-model-no-known}). By definition, the known covariates correspond to measured or ``natural'' variables, and hence they have no such rotational symmetry.

Using standard mixed-model calculations to integrate out the random
effects (SI Section~\ref{sec:model}), the log-likelihood of the unknown
model parameters given the observed data can be written as
\begin{equation}\label{eq:2}
  \L\bigl(\X,\A,\B,\sigma^2\mid \Y, \Z) = - \log\det(\K) - \tr\bigl(\K^{-1} \C),
\end{equation}
where
\begin{equation}\label{eq:3}
  \K = \Z\B\Z^T + Z\D\X^T + \X\D^T\Z^T + \X\A\X^T + \sigma^2 \I
\end{equation}
and $\C=(\Y\Y^T)/m$ is the sample covariance matrix. Maximizing the log-likelihood (\ref{eq:2}) over positive definite matrices $\K$ without any further constraints would result in the estimate $\hat\K=\C$ (note that $\C$ is invertible because we assume that the number of genes $m$ is greater than the number of samples $n$) \cite{anderson1985maximum}.

 If $\K$ is constrained to be of the form $\K=\X\A\X^T + \sigma^2 \I$ for a given number of latent factors $p<n$, then the model is known as probabilistic PCA, and the likelihood is maximized by identifying the latent factors with the eigenvectors of $\C$ corresponding to the $p$ largest eigenvalues \cite{tipping1999probabilistic,lawrence2005probabilistic}. In matrix form, the probabilistic PCA solution can be written as
\begin{equation}\label{eq:37}
  \hat\K = \P_1\C\P_1 + \hat\sigma^2 \P_2,
\end{equation}
 where $\P_1$ and $\P_2$ are mutually orthogonal projection matrices on the space spanned by the first $p$ and last $n-p$ eigenvectors of $\C$, respectively, and the maximum-likelihood estimate $\hat\sigma^2$ is the average variance explained by the $n-p$ excluded dimensions (SI Section~\ref{sec:solut-model-no-known}). 

If $\K$ is constrained to be of the form $\K=\Z\B\Z^T + \sigma^2 \I$, the model is a standard random effect model with the same design matrix $\Z$ for the random effects $v_i$ for each gene $i$. In general, there exists no analytic solution for the maximum-likelihood estimates of the (co)variance parameter matrix $\B$  in a random effect model \cite{gumedze2011parameter}. However in the present context it is assumed that the data for each gene are an independent sample of the \emph{same} random effect model. Again using the fact that  $\C=(\Y\Y^T)/m$ is invertible due to the number of genes being greater than the number of samples, the maximum-likelihood solution for $\B$, and hence $\K$, can be found analytically in terms of $\C$ and the singular value decomposition (SVD) of $\Z$. It turns out to be of the same form (\ref{eq:37}), except that $\P_1$ now projects onto the subspace spanned by the known covariates (the columns of $\Z$) (SI Section \ref{sec:solut-model-no-latent}).

In the most general case where $\K$ takes the form (\ref{eq:3}), we show first that every model of the form (\ref{eq:38}) can be rewritten as a model of the same form where the hidden factors are orthogonal to the known covariates, $\X^T\Z=0$. The reason is that any overlap between the hidden and known covariates can be absorbed in the random effects $v_i$ by a linear transformation, and therefore simply consists of a reparameterization of the covariance matrices $\B$ and $\D$ (SI Section~\ref{sec:solution-full-model}). Once this orthogonality is taken into account, the log-likelihood (\ref{eq:2}) decomposes as a sum $\L=\L_1+\L_2$, where $\L_2$ is identical to the log-likelihood of probabilistic PCA on the \emph{reduced space} that is the orthogonal complement to the subspace spanned by the known covariates (columns of $\Z$). Analogous to the restricted maximum-likelihood (REML) method for ordinary linear mixed models, where variance parameters of the random effects are estimated in the subspace orthogonal to the maximum-likelihood estimates of the fixed effects \cite{patterson1971recovery,gumedze2011parameter}, we estimate the latent variables $\X$ by maximizing only the likelihood term $\L_2$ corresponding to the subspace where these $\X$ live (Section~\ref{sec:solution-full-model}). Once the restricted maximum-likelihood estimates $\hat\X$ are determined, they become ``known'' covariates, allowing the covariance parameter matrices to be determined by maximizing the remaining terms $\L_1$ in the likelihood function using the analytic solution for a model with known covariates $(\Z,\;\hat\X)$ (Section~\ref{sec:solution-full-model}).

By analogy with the REML method, we call our method the restricted maximum-likelihood method for solving the latent variable model (\ref{eq:38}), abbreviated ``\lvreml''.  While the \lvreml\ solution is not guaranteed to be the absolute maximizer of the total likelihood function, it is guaranteed analytically that for any given number $p$ of latent variables, the \lvreml\ solution attains minimal unexplained variance among all possible choices of $p$ latent variables (SI Section~\ref{sec:solution-full-model}). 

\subsection{\lvreml, a flexible software package for learning latent variance components in gene expression data}
\label{sec:lvreml-flex-softw}

\begin{figure}
  \centering
  \includegraphics[width=\linewidth]{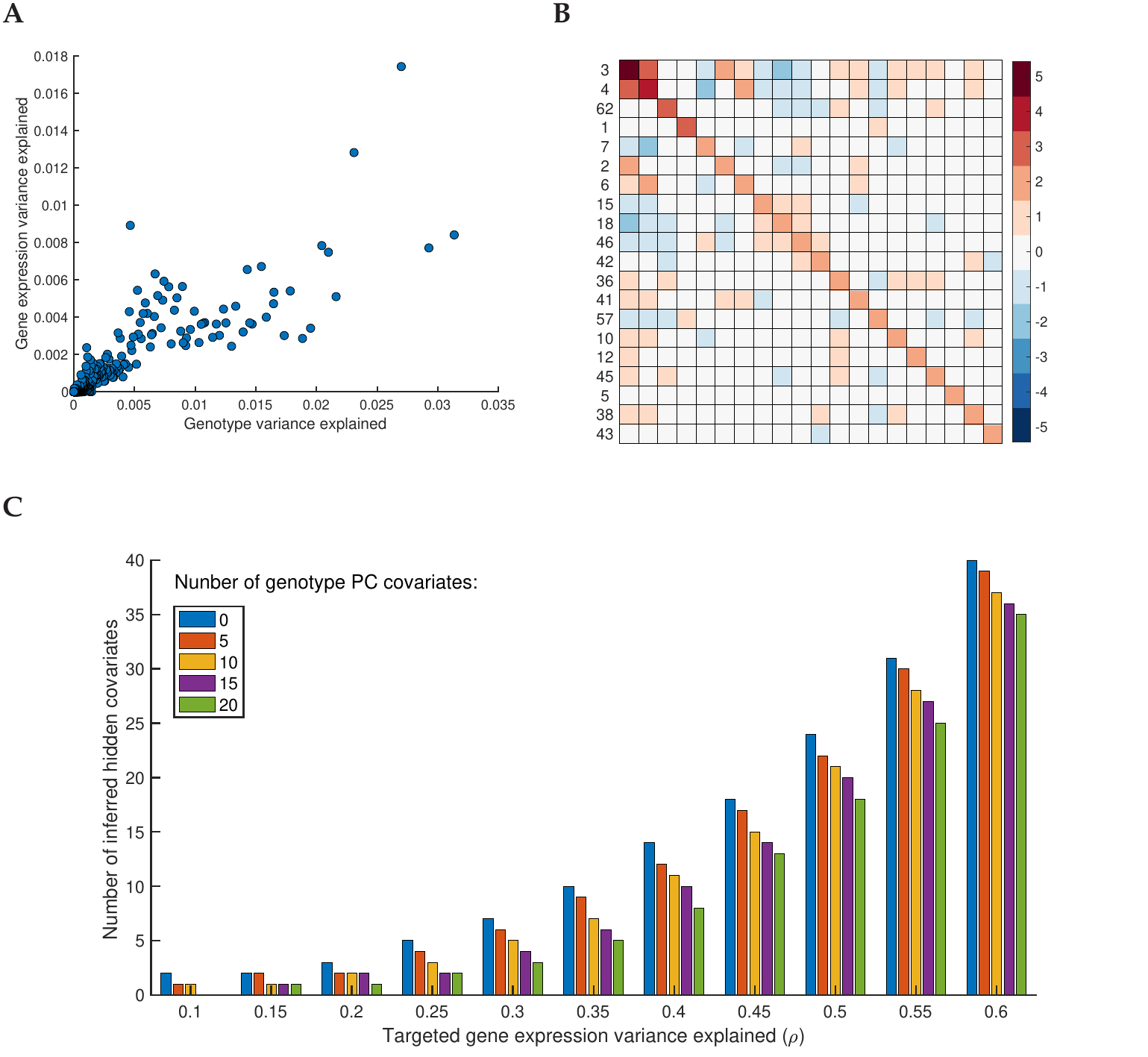}
  \caption{\textbf{A.} Gene expression variance explained by individual genotype PCs in univariate models vs.\ their genotype variance explained. \textbf{B.} Heatmap of the estimated covariance matrix $\B$ (cf.\ eq.~\eqref{eq:3}) among the effects on gene expression of the top 20 genotype PCs (by gene expression variance explained in univariate models, cf.\ panel A, y-axis); the row labels indicate the genotype PC index, ranked by genotype variance explained (cf.\ panel A, x-axis).  \textbf{C.} Number of hidden covariates inferred by \lvreml\ as a function of the parameter $\rho$ (the targeted total amount of variance explained by the known and hidden covariates), with $\theta$ (the minimum variance explained by a known covariate) set to retain 0, 5, 10, or 20 known covariates (genotype PCs) in the model. For visualization purposes only the range of $\rho$ upto $\rho=0.6$ is shown, for the full range, see Supp.\ Fig.\ \ref{fig:num_hidden_rho_supp}.}
  \label{fig:numcovar-theta}
\end{figure}

We implemented the restricted maximum-likelihood method for solving model (\ref{eq:38}) in a software package \lvreml, available with Matlab and Python interfaces at \lvurl. \lvreml\ takes as input a gene expression matrix $\Y$, a covariate matrix $\Z$, and a parameter $\rho$, with $0<\rho<1$. This parameter is the desired proportion of variation in $\Y$ that should be explained by the combined known and latent variance components. Given $\rho$, the number of latent factors $p$ is determined automatically (SI Section~\ref{sec:select-covar-latent}). \lvreml\ centres the data $\Y$ such that each sample has mean value zero, to ensure that no fixed effects on the mean need to be included in the model (SI Section~\ref{sec:syst-effects-mean}).

When the number of known covariates (or more precisely the rank of $\Z$) exceeds the number of samples, as happens in eQTL studies where a large number of SNPs can act as covariates \cite{fusi2012joint}, a subset of $n$ linearly independent covariates will always  explain \emph{all} of the variation in $\Y$. In \cite{fusi2012joint}, a heuristic approach was used to select covariates during the likelihood optimization, making it difficult to understand \textit{a priori} which covariates will be included in the model and why. In contrast, \lvreml\ includes a function to perform initial screening of the covariates, solving for each one the model (\ref{eq:38}) with a single known covariate to compute the variance $\hat\beta^2$ explained by that covariate alone (SI Section~\ref{sec:solut-model-no-latent}). This estimate is then used to include in the final model only those covariates for which $\hat\beta^2\geq\theta\tr(\C)$, where  $\theta>0$ is the second free parameter of the method, namely the minimum amount of variation  a known covariate needs to explain on its own to be included in the model (SI Section~\ref{sec:select-covar-latent}). In the case of genetic covariates, we further propose to apply this selection criterion not to individual SNPs, but to principal components (PCs) of the genotype data matrix. Since PCA is a linear transformation of the genotype data, it does not alter model (\ref{eq:38}). Moreover, selecting PCs as covariates ensures that the selected covariates are linearly independent, and is consistent with the fact that genotype PCs are known to reveal population structure in expression data \cite{brown2018expression}.

To test  \lvreml\ and illustrate the effect of its parameters, we used genotype data for 42,052 genetic markers and RNA sequencing expression data for 5,720 genes in 1,012 segregants from a cross between two strains of budding yeast \cite{albert2018genetics}, one of the largest (in terms of sample size),  openly available eQTL studies in any organism (Methods Section~\ref{sec:data}). We first performed PCA on the  genotype data.  The dominant genotype PCs individually explained 2-3\% of variation in the genotype data, and 1-2\% of variation in the expression data, according to the single-covariate model (SI Section~\ref{sec:solut-model-no-latent}, eq.~(\ref{eq:34})) (Fig.~\ref{fig:numcovar-theta}A). Although genotype PCs are orthogonal by definition, their effects on gene expression are not independent, as shown by the non-zero off-diagonal entries in the maximum-likelihood estimate of the covariance matrix $\B$ (cf.\ \eqref{eq:3}) (Fig.~\ref{fig:numcovar-theta}B). To illustrate how the number of inferred hidden covariates varies as a function of the input parameter $\rho$, we determined values of the parameter $\theta$ to include between 0 and 20 genotype PCs as covariates in the model. As expected, for a fixed number of known covariates, the number of hidden covariates increases with $\rho$, as more covariates are needed to explain more of the variation in $\Y$, and decreases with the number of known covariates, as fewer hidden covariates are needed when the known covariates already explain more of the variation in $\Y$ (Fig.~\ref{fig:numcovar-theta}C).

When setting the parmater $\theta$, or equivalently, deciding the number of known covariates to include in the model, care must be taken due to a mathematical property of the model: the maximizing solution exists only if the minimum amount of variation in $\Y$ explained by a known covariate (or more precisely, by a principal axis in the space spanned by the known covariates) is greater than the maximum-likelihood estimate of the residual variance $\hat\sigma^2$ (see Theorems~\ref{thm:no-latent} and \ref{thm:full} in SI Sections~\ref{sec:solut-model-no-latent} and \ref{sec:solution-full-model}).  If non-informative variables are included among the known covariates, or known covariates are strongly correlated, then the minimum variation explained by them becomes small, and potentially smaller than the residual variance, whose initial ``target'' value is $1-\rho$. Because \lvreml\ considers the known covariates as fixed, it lowers the value of $\hat\sigma^2$ by including more hidden covariates in the model, until the existence condition is satisfied. In such cases, the total variance explained by the known and hidden covariates will be greater than the target value of the input parameter $\rho$. Visually, the presence of non-informative dimensions in the linear subspace spanned by the known covariates (due to non-informative or redundant variables) is shown by a saturation of the number of inferred hidden covariates with decreasing $\rho$ (Supp.\ Fig.\ \ref{fig:num_hidden_rho_supp}B), providing a clear cue that the relevance or possible redundancy of (some of) the known covariates for explaining variation in the expression data needs to be reconsidered.

\subsection{\lvreml\  attains likelihood values higher than or equal to \panama}
\label{sec:panama-hidd-fact}

To compare the analytic solution of \lvreml\ against the original model with gradient-based optimization algorithm, as implemented in the \panama\ software \cite{fusi2012joint}, we performed a controlled comparison where  0, 5, 10, and 20 dominant principal components (PCs) of the  expression data $\Y$ were used as artificial known covariates. Because of the mathematical properties of the model and the \lvreml\ solution, if the first $d$ expression PCs are included as known covariates, \lvreml\ will return the next $p$ expression PCs as hidden factors. Hence the log-likelihood of the \lvreml\ solution with $d$ expression PCs as known covariates and $p$ hidden factors will coincide with the log-likelihood of the solution with zero known covariates and $d+p$ hidden factors (that is, probabilistic PCA with  $d+p$ hidden factors). Fig.~\ref{fig:log-like}A shows that this is the case indeed: the log-likelihood curves for 0, 5, 10, and 20 PCs as known covariates are shifted horizontally by a difference of exactly 5  (from 0, to 5, to 10) or 10 (from 10 to 20) hidden factors.

\begin{figure}
  \centering
  \includegraphics[width=\linewidth]{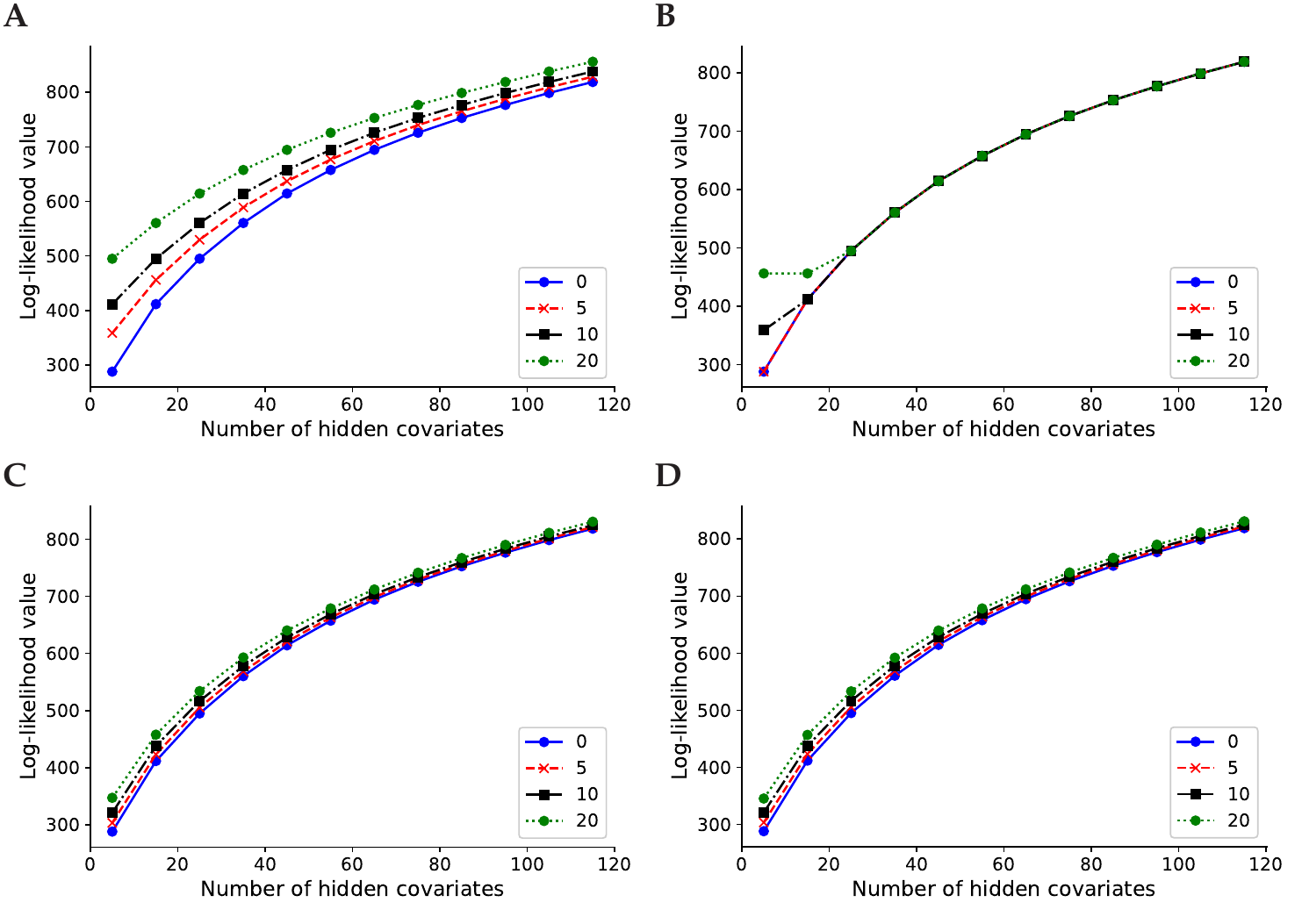}
  \caption{Log-likelihood values for \lvreml\ (\textbf{A,C}) and \panama\ (\textbf{B,D}) using 0, 5, 10, and 20 principal components of the expression data (\textbf{A},\textbf{B}) or genotype data (\textbf{C},\textbf{D}) as known covariates. The results shown are for 600 randomly subsampled segregants; corresponding results for 200, 400, and in the case of \lvreml\, 1,012 segregants are shown in Supp.\ Fig.\ \ref{fig:likelihood_supp}.}
  \label{fig:log-like}
\end{figure}

In contrast, \panama\ did not find the optimal shifted probabilistic PCA solution, and its likelihood values largely coincided with the solution with zero known covariates, irrespective of the number of known covariates provided (Fig.~\ref{fig:log-like}B). In other words, \panama\ did not use the knowledge of the known covariates to explore the orthogonal space of axes of variation not yet explained by the known covariates, instead arriving at a solution where $p$ hidden factors appear to explain no more of the variation than $p-d$ PCs orthogonal to the $d$ known PCs. To verify this, we compared the \panama\ hidden factors to PCs given as known covariates, and found that in all cases where the curves in Fig.~\ref{fig:log-like}B align, the first $d$ hidden factors coincided indeed with the $d$ known covariates (data not shown).

When genotype PCs were used as known confounders (using the procedure explained in Section~\ref{sec:lvreml-flex-softw}), the shift in log-likelihood values was less pronounced, consistent with the notion that the genotype PCs explain less of the expression variation than the expression PCs. In this case the likelihood values of \lvreml\ and \panama\ coincided  (Fig.~\ref{fig:log-like}C,D), indicating that both methods found the same optimal covariance matrix. 

The explanation for the difference between Fig.~\ref{fig:log-like}A and C is as follows. In Fig.~\ref{fig:log-like}A,  \lvreml\ uses $p$ hidden covariates to explain the same amount of variation as $d+p$ expression PCs. The dominant expression PCs are partially explained by population structure (genotype data). Hence when $d$ genotype PCs are given as known covariates, \lvreml\  infers $p$ orthogonal latent variables that explain the ``missing'' portions of the expression PCs not explained by genotype data. This results in a model that explains more expression variation than the $p$ dominant expression PCs, but less than $p+d$ expression PCs, hence the reduced shift in Fig.~\ref{fig:log-like}C. 

It is unclear why \panama\ did not find the correct solution when expression PCs were used as known covariates (Fig.~\ref{fig:log-like}B), but this behaviour was consistent across multiple subsampled datasets of varying sizes (Supp.\ Fig.\ \ref{fig:likelihood_supp}) as well as in other datasets (data not shown).

\subsection{\panama\ and \peer\ infer hidden factors that are partially redundant with the known covariates}

Although \panama\ inferred models with the same covariance matrix estimate $\hat \K$ (cf.\ Section \ref{sec:exact-solut-line}) and hence the same likelihood values as \lvreml\ when genotype PCs where given as known covariates, the inferred hidden covariates differed between the methods.

\begin{figure}
  \centering
 \includegraphics[width=\linewidth]{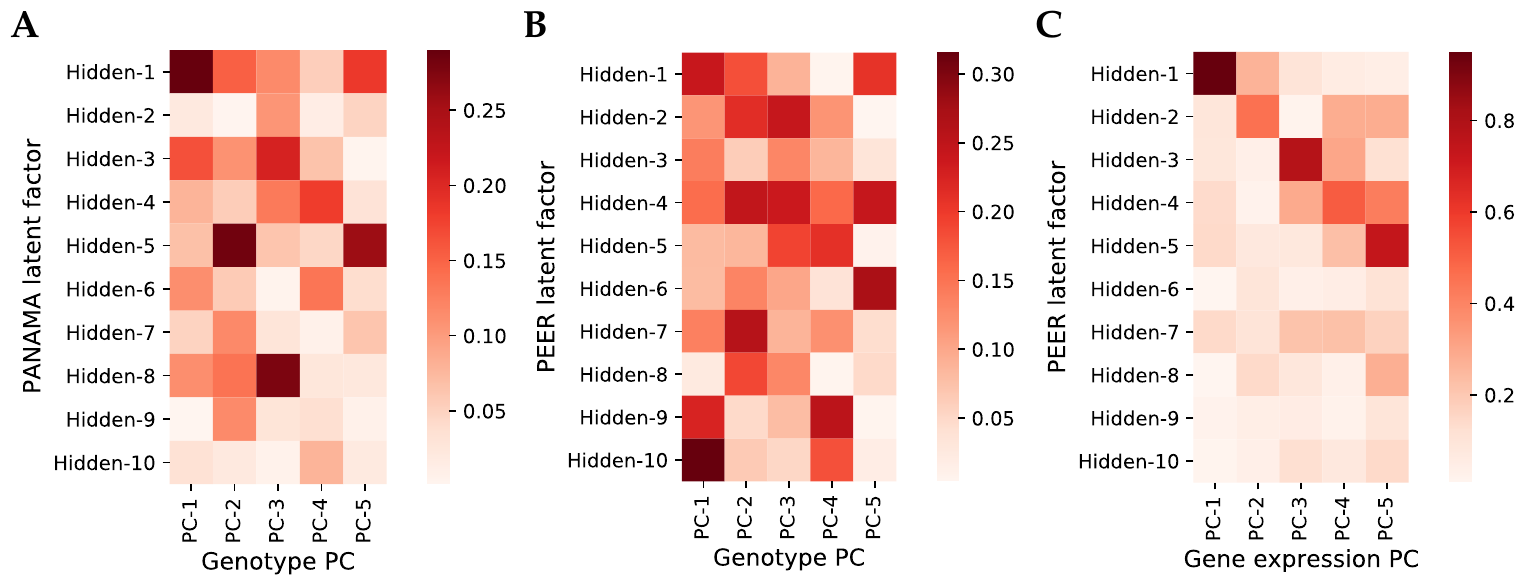}
 \caption{Cosine similarity between known covariates (5 genotype PCs) given to the model and hidden factors inferred by \panama\ (\textbf{A}) and \peer\ (\textbf{B}), and cosine similarity between gene expression PCs and and hidden factors inferred by \peer\ (\textbf{C}) when no known covariates are given to the model. Results are for randomly subsampled data of 200 segregants.}
  \label{fig:overlap-snp}
\end{figure}

As explained, hidden covariates inferred by \lvreml\ are automatically orthogonal to the known covariates (cf.\ Section \ref{sec:exact-solut-line}) and represent linearly independent axes of variation. In contrast, the latent variables inferred by \panama\ overlapped with the known genotype covariates supplied to the model, with cosine similarities of upto 30\% (Figure \ref{fig:overlap-snp}A). In  \panama\, covariances among the effects of the known confounders are assumed to be zero. When the optimal model (i.e.\ maximum-likelihood $\hat \K$) in fact has effects with non-zero covariance (as in Fig.~\ref{fig:numcovar-theta}B), the optimization algorithm in \panama\ will automatically select hidden confounders that overlap with the known confounders in order to account for these non-zero covariances (SI Section~\ref{sec:solution-full-model}), thus resulting in the observed overlap. Hence the common interpretation of \panama\ factors as new determinants of gene expression distinct from known genetic factors is problematic.

To test whether the overlap between inferred and already known covariates also occurs in other methods or is specific to \panama, we ran the \peer\ software \cite{stegle2012using} on a reduced dataset of 200 randomly selected samples from the yeast data (\peer\ runtimes made it infeasible to run on larger sample sizes). \peer\ is a popular software that uses a more elaborate hierarchical model to infer latent variance components \cite{stegle2010bayesian}. \peer\ hidden factors again showed cosine similarities of upto 30\% (Figure \ref{fig:overlap-snp}B), suggesting that its hidden factors also cannot be interpreted as completely new determinants of gene expression. We also tested the hidden factors returned by \peer\ when no known covariates are added to the model. In this case, model \eqref{eq:38} reduces to probabilistic PCA and both \lvreml\ and \panama\ correctly identify the dominant expression PCs as hidden factors (Fig.\ \ref{fig:log-like}A,B). Despite its more complex model, which does not permit an analytic solution even in the absence of known covariates, \peer\ hidden factors in fact do overlap strongly with the same dominant expression PCs (cosine similarities between 60\% to 80\%), indicating that the added value of the more complicated model structure may be limited, at least in this case.

\subsection{\lvreml\ is orders of magnitude faster than \panama}

An analytic solution does not only provide additional insight into the mathematical properties of a model, but can also provide significant gains in computational efficiency. The \lvreml\ solution can be computed using standard matrix operations from linear algebra, for which highly optimized implementations exist in all programming languages. Comparison of the runtime of the Python implementations of \lvreml\ and \panama\ on the yeast data at multiple sample sizes showed around ten thousand-fold speed-up factors, from several minutes for a single \panama\ run to a few tens of milliseconds for \lvreml\ (Fig.~\ref{fig:runtime}). Interestingly, the computational cost of \lvreml\ did not increase much when known covariates were included in the model, compared to the model without known covariates that is solved by PCA (Fig.~\ref{fig:runtime}A). In contrast, runtime of \panama\ blows up massively as soon as covariates are included (Fig.~\ref{fig:runtime}B). Nevertheless, even in the case of no covariates, \panama\ is around 600 times slower than the direct, eigenvector decomposition based solution implemented in \lvreml. Finally, the runtime of \lvreml\ does not depend on the number of known or inferred latent factors, whereas increasing either parameter in \panama\ leads to an increase in runtime (Supp.\ Fig.~\ref{fig:runtime_supp}). 

\begin{figure}
  \centering
  \includegraphics[width=\linewidth]{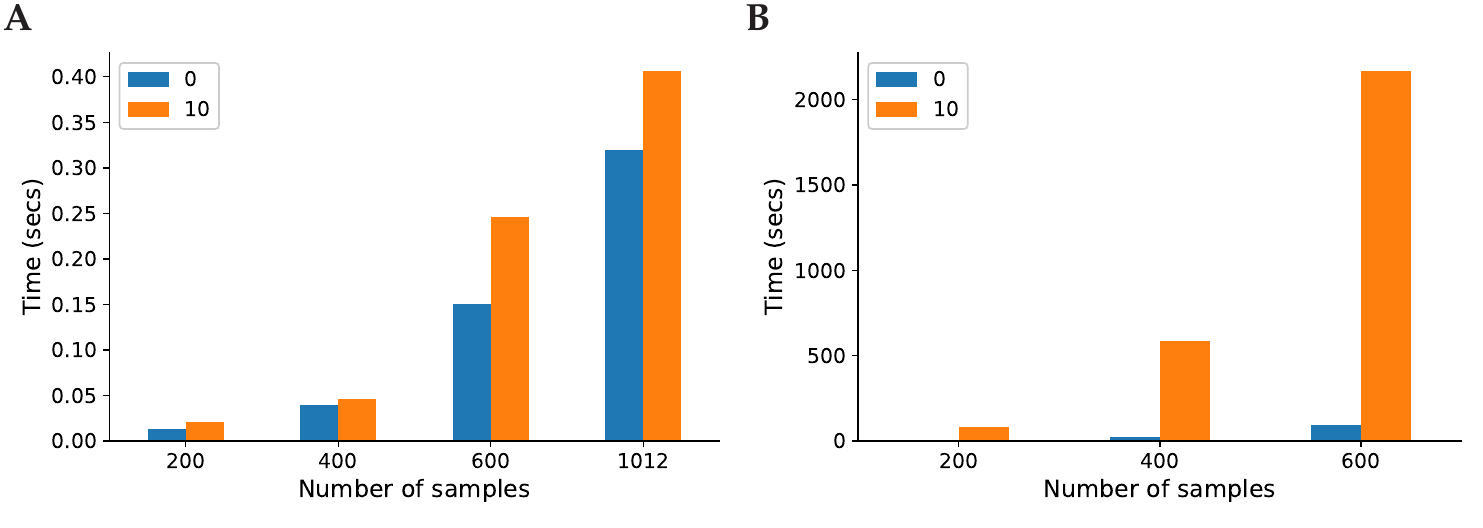}
  \caption{Runtime comparison on between \lvreml\ (\textbf{A}) and \panama\ (\textbf{B}), with parameters set to infer 85 hidden covariates with either 0 known covariates or including 10 genotype PCs as known covariates, at multiple sample sizes. Running \panama\ on the full dataset of 1,012 segregants was infeasible. For runtime comparisons at other parameter settings, see Supp.\ Fig.\ \ref{fig:runtime_supp}.}
  \label{fig:runtime}
\end{figure}

\section{Discussion}
\label{sec:discussion}

We presented a random effect model to estimate simultaneously the contribution of known and latent variance components in gene expression data, that is closely related to models that have been used previously in this context \cite{lawrence2005probabilistic, stegle2010bayesian, stegle2012using,  fusi2012joint,buettner2015computational}. By including additional parameters in our model to account for non-zero covariances among the effects of known covariates and latent factors, we were able to show that latent factors can always be taken orthogonal to, and therefore linearly independent of, the known covariates supplied to the model. This is important, because inferred latent factors are not only used to correct for correlation structure in the data, but also as new, data-derived ``endophenotypes'', that is, determinants of gene expression whose own genetic associations are biologically informative \cite{parts2011joint, stegle2012using}. As shown in this paper, the existing models and their numerical optimization result in hidden factors that in fact overlap significantly with the known covariates, and hence their value in uncovering ``new'' determinants of gene expression must be questioned.

To solve our model, we did not rely on numerical, gradient-based optimizers, but rather on an analytic restricted maximum-likelihood solution. This solution relies on a decomposition of the log-likelihood function that allows to identify hidden factors as principal components of the expression data matrix reduced to the orthogonal complement of the subspace spanned by the known covariates. This solution is guaranteed to minimize the amount of unexplained variation in the expression data for a given number of latent factors, and is analogous to the widely used restricted maximum-likelihood solution for conventional linear mixed models, where variance parameters of random effects are estimated in the subspace orthogonal to the maximum-likelihood estimates of the fixed effects.

Having an analytic solution is not only important for understanding the mathematical properties of a statistical model, but can also lead to significant reduction of the computational cost for estimating parameter values. Here we obtained a ten-thousand fold speed-up compared to an existing software that uses gradient-based optimization. On a yeast dataset with 1,012 samples, our method could solve the covariance structure and infer latent factors in less than half a second, whereas it was not feasible to run an existing implementation of gradient-based optimization on more than 600 samples.

The experiments on the yeast data showed that in real-world scenarios, \lvreml\ and the gradient-based optimizer implemented in the \panama\ software resulted in the same estimates for the sample covariance matrix. Although the latent variables inferred by both methods are different (orthogonal vs.\ partially overlapping with the population structure covariates), we anticipate that downstream linear association analyses will nevertheless give similar results as well. For instance, established protocols\cite{stegle2012using} recommend to use known and latent factors as covariates to increase the power to detect expression QTLs. Since orthogonal and overlapping latent factors can be transformed into each other through a linear combination with the known confounders, linear assocation models that use both known and latent factors as covariates will also be equivalent (SI Section~\ref{sec:downstream-analyses}).

While we have demonstrated that the use of latent variance components that are orthogonal to known confounders leads to significant analytical and numerical advantages, we acknowledge that it follows from a mathematical symmetry of the underlying statistical model that allows to transform a model with overlapping latent factors to an equivalent model with orthogonal factors. Whether the true but unknown underlying variance components are orthogonal or not, nor their true overlap value with the known confounders, can be established by the models studied in this paper precisely due this mathematical symmetry. Such limitations are inherent to all latent variable methods.

To conclude, we have derived an analytic restricted maximum-likelihood solution for a widely used class of random effect models for learning latent variance components in gene expression data with known and unknown confounders. Our solution can be computed in a highly efficient manner, identifies hidden factors that are orthogonal to the already known variance components, and results in the estimation of a sample covariance matrix that can be used for the downstream estimation of variance parameters for individual genes. The restricted maximum-likelihood method  facilitates the application of random effect modelling strategies for learning latent variance components to much larger gene expression datasets than currently possible.

\section{Methods}
\label{sec:methods}

\subsection{Mathematical methods}
\label{sec:mathematical-methods}

All model equations, mathematical results and detailed proofs are described in a separate Supplementary Information document.

\subsection{Data}
\label{sec:data}

We used publicly available genotype and RNA sequencing data from 1,012 segregants from a cross between two yeast strains \cite{albert2018genetics}, consisting of gene expression levels for 5,720 genes and (binary) genotype values for 42,052 SNPs. Following \cite{albert2018genetics}, we removed batch and optical density effects from the expression data using categorical regression. The expression residuals were centred such that each sample had mean zero to form the input matrix $\Y$ to the model (cf.\ SI Section \ref{sec:model}). L2-normalized genotype PCs were computed using the singular value decomposition of the genotype data matrix with centred (mean zero) samples, and used to form input matrices $\Z$ to the model (cf.\ SI Section \ref{sec:model}). Data preprocessing scripts are available at \lvurl.

\subsection{\lvreml\ analyses}
\label{sec:lvreml-analyses}

The \lvreml\ software, as well as a script that details the \lvreml\ analyses of the yeast data is available at \lvurl.

\subsection{\panama\ analyses}
\label{sec:panama-analyses}

We obtained the \panama\ software from the \limix\ package available at \panurl.

The following settings were used to ensure that exactly the same normalized data was used by both methods: 1) For parameter \textbf{Y}, the same gene expression matrix, with each sample normalized to have zero mean, was used that was used as input for \lvreml, setting the \textbf{standardize} parameter to \texttt{false}. 2) The parameter \textbf{Ks} requires a list of covariance matrices for each known factor. Therefore for each column $z_i$ of the matrix \textbf{Z} used by \lvreml\, we generated a covariance matrices $\mathbf{Ks}_i = z_i z_i^T$. The `\textbf{use Kpop}' parameter, which is used to supply a population structure covariance matrix to \panama\ in addition to the known covariates, was set to \texttt{false}.  

To be able to calculate the log-likelihoods and extract other relevant information from the \panama\ results, we made the following modifications to the \panama\ code: 1) The covariance matrices returned by \panama\ are by default normalized by dividing the elements of the matrix by the mean of its diagonal elements.  To make these covariance matrices comparable to \lvreml\, this normalization was omitted by commenting out the lines in the original \panama\ code where this normalization was being performed. 2) \panama\ does not return the variance explained by the known confounders unless the `\textbf{use Kpop}' parameter is set to \texttt{true}. 
Therefore the code was modified so that it would still return the variance explained by the known confounders. 3) The \textbf{K} matrix returned by \panama\ does not include the effect of the noise parameter $\sigma^2$. Therefore the code was modified to return the {$\sigma^2 \I$ matrix, which was then added to the returned \textbf{K} i.e. $\K_{new} = \K + \sigma^2 \I$, to be able to use equation (\ref{eq:2}) to compute the log-likelihood.  The modified code is available as a fork of the \textsc{limix} package at \url{https://github.com/michoel-lab/limix-legacy}





\newpage


\renewcommand\thesection{S\arabic{section}}
\renewcommand\thefigure{S\arabic{figure}}
\renewcommand\thetable{S\arabic{table}}
\renewcommand\theequation{S\arabic{equation}}
\setcounter{figure}{0}
 \setcounter{table}{0}
 \setcounter{section}{0}
\setcounter{equation}{0}

\chapter*{\centering Supplementary Figures}

\begin{figure}[h!]
  \includegraphics[width=\linewidth]{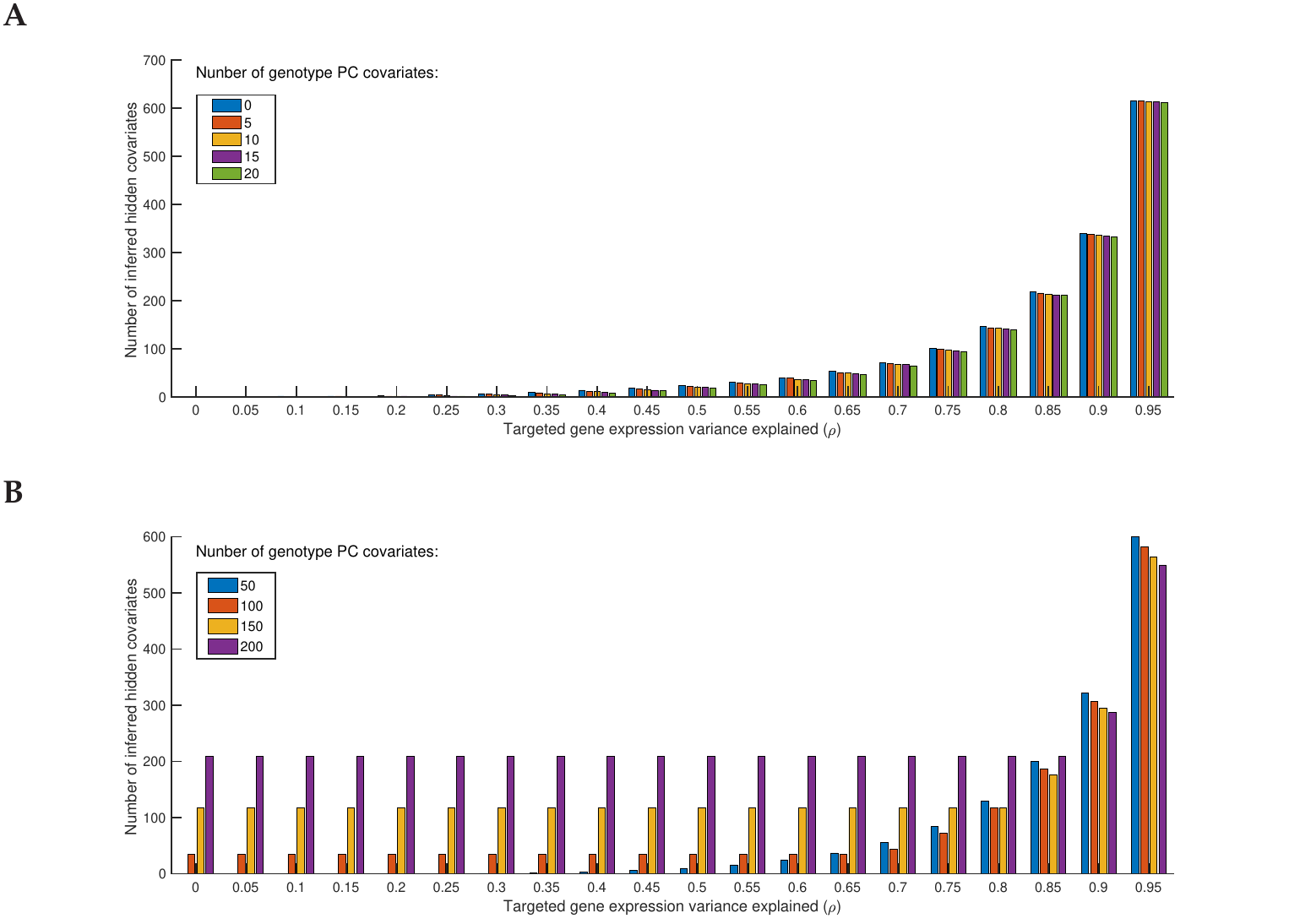}
  \caption{\textbf{A.} Number of hidden covariates inferred by \lvreml\ as a function of the parameter $\rho$ (the targeted total amount of variance explained by the known and hidden covariates), with $\theta$ (the minimum variance explained by a known covariate) set to retain 0, 5, 10, or 20 known covariates (genotype PCs) in the model. \textbf{B.} Same as panel \textbf{A}, with $\theta$ set to retain 50, 100, 150, or 200 genotype PCs in the model. The saturation of the number of hidden covariates with decreasing $\rho$ for models with 100, 150, and 200 known covariates is a visual indicator that some of the dimensions in the linear subspace spanned by the known covariates do not explain sufficient variation in the expression data, and the relevance or possible redundancy of (some of) the known covariates for explaining variation in the expression data needs to be reconsidered.}
  \label{fig:num_hidden_rho_supp}
\end{figure}

\begin{figure}[h!]
  \centering
  \includegraphics*[width=.63\linewidth]{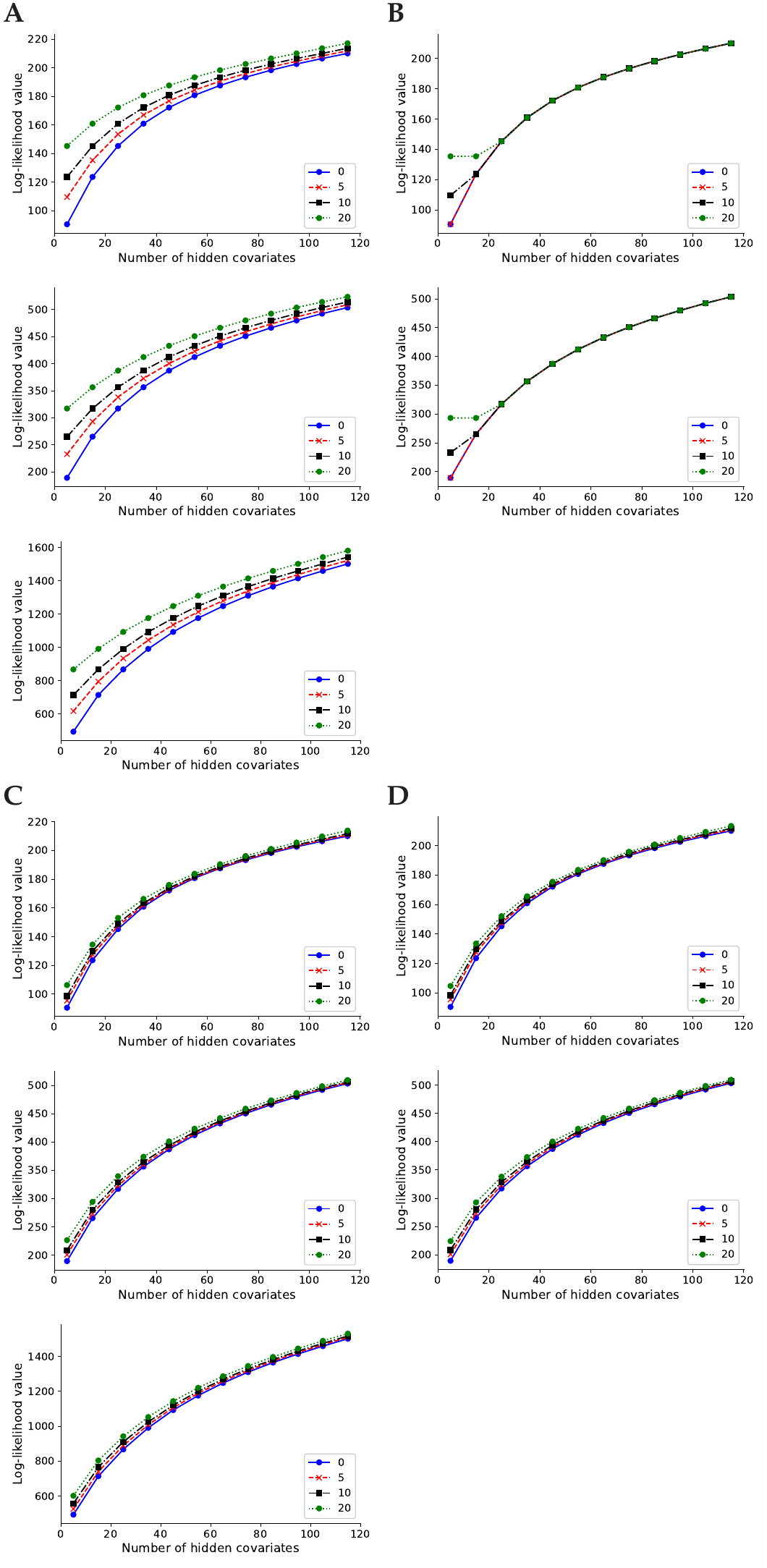}
  \caption{Log-likelihood values for \lvreml\ (\textbf{A,C}) and \panama\ (\textbf{B,D}) using 0, 5, 10, and 20 PCs of the expression data (\textbf{A},\textbf{B}) or genotype data (\textbf{C},\textbf{D}) as known covariates, at sample sizes of 200, 400, and in the case of \lvreml\, 1,012 segregants (top to bottom).}
  \label{fig:likelihood_supp}
\end{figure}

\begin{figure}[h!]
  \centering
  \includegraphics[width=.95\linewidth]{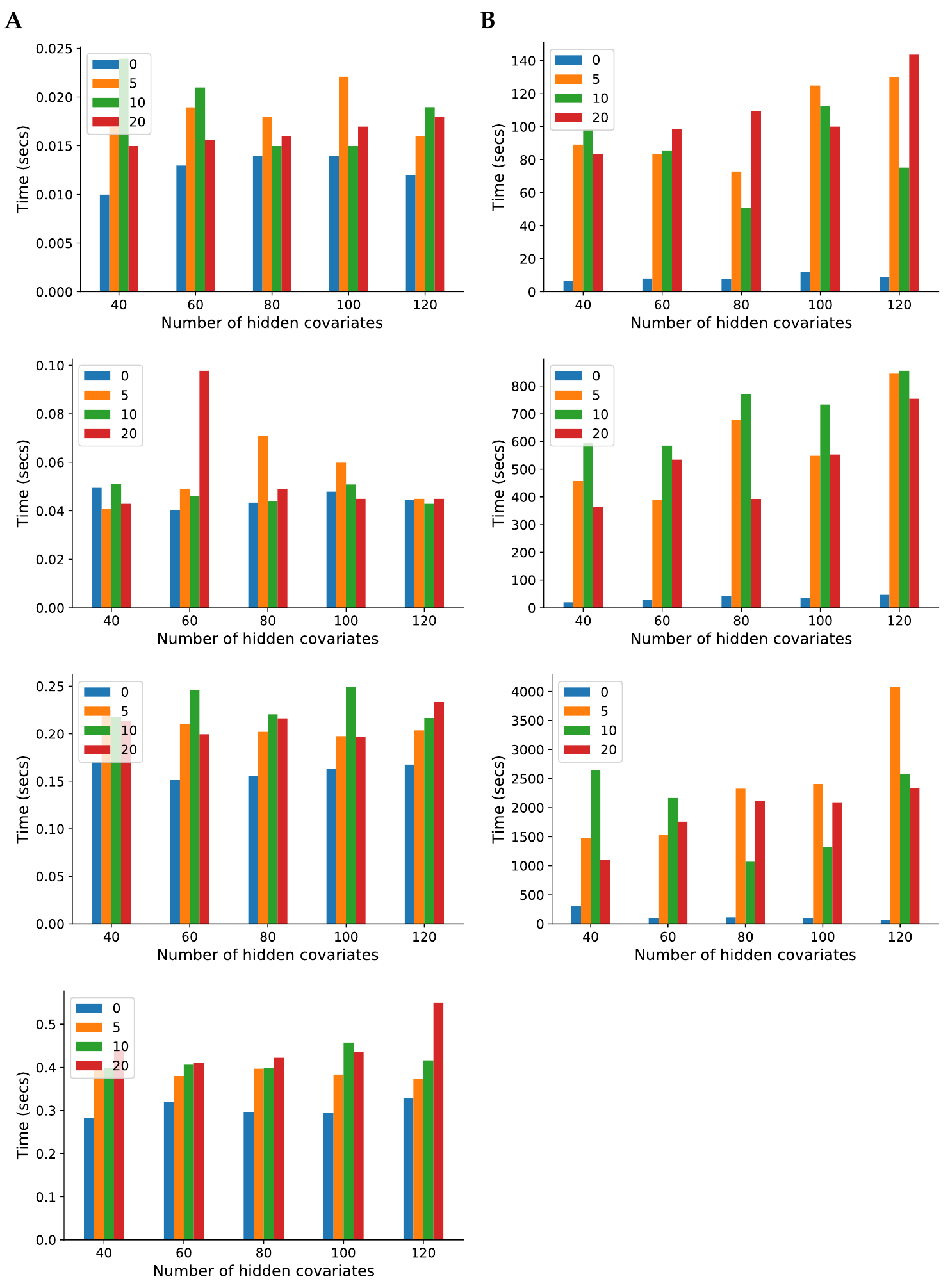}
  \caption{Runtime comparison on between \lvreml\ (\textbf{A}) and \panama\ (\textbf{B}), with parameters set to infer 85 hidden covariates with 0, 5, 10, or 20 genotype PCs included as known covariates, at  sample sizes of 200, 400, and in the case of \lvreml\, 1,012 segregants (top to bottom).}
  \label{fig:runtime_supp}
\end{figure}

\cleardoublepage


\chapter*{\centering Supplementary Methods}

\section{Preliminary results}
\label{sec:preliminary-results}

In the sections below, we will repeatedly use the following
results. The first result concerns linear transformations of normally distributed variables and can be found in most textbooks on statistics
or probability theory:
\begin{lem}\label{lem:affine-gaussian}
  Let $x\in\R^n$ be a random, normally distributed vector,
  \begin{align*}
    p(x) = \N(\mu,\Psi),
  \end{align*}
  with $\mu\in\R^n$, and $\Psi\in\R^{n\times n}$ a positive definite
  covariance matrix. For any linear transformation $y=\M x$ with
  $\M\in\R^{m\times n}$, we have
  \begin{align*}
    p(y) = \N(\M\mu,\M\Psi\M^T).
  \end{align*}\qed
\end{lem}

If the linear transformation $y=\M x$ in this Lemma is overdetermined, that is, if $m>n$, then the transformed covariance matrix $\Psi'=\M\Psi\M^T$ will have a lower rank $n$ than its dimension $m$, that is, $\Psi'\in\R^{m\times m}$ is a positive \emph{semi}-definite matrix (i.e., has one or more zero eigenvalues). Thus we can extend the definition of normal distributions to include \emph{degenerate}  distributions with  positive \emph{semi}-definite covariance matrix, by interpreting them as the distributions of  overdetermined linear combinations of normally distributed vectors. A degenerate one-dimensional normal distribution is simply defined as a $\delta$-distribution, that is, for $x\in\R$
\begin{align*}
  p(x) = \N(\mu,0) = \delta(x-\mu),
\end{align*}
which can be derived as a limit $\sigma^2\to 0$ of normal distribution density functions $\N(\mu,\sigma^2)$.

The second result is one that is attributed to von Neumann \cite{anderson1985maximum}:
\begin{lem}\label{lem:von-Neumann}
  Let $\P,\Q\in\R^{n\times n}$ be two positive definite matrices. Then
  \begin{equation}\label{eq:25}
    \tr(\P^{-1}\Q) \geq \sum_{i=1}^n \pi^{-1}_i \chi_i,
  \end{equation}
  where $\pi_1\geq \dots \geq \pi_n$ and $\chi_1\geq\dots\geq \chi_n$
  are the ordered eigenvalues of $\P$ and $\Q$, respectively, and
  equality in eq.~\eqref{eq:25} is achieved if and only if the
  eigenvector of $\P$ corresponding to $\pi_i$ is equal to the
  eigenvector of $\Q$ corresponding to $\chi_{n-i+1}$,
  $i=1,\dots,n$. \qed
\end{lem}

\section{The model}
\label{sec:model}

We will use the following notation:
\begin{itemize}
\item $\Y\in\R^{n\times m}$ is a matrix of gene expression data for $m$ genes in $n$ samples. The $i$th column of $\Y$ is denoted $y_i\in\R^n$ and corresponds to the vector of expression values for gene $i$. We assume that the data in each sample are centred, $\sum_{i=1}^m y_i=0\in\R^{n}$.
\item $\Z\in\R^{n\times d}$ is a matrix of values for $d$ known confounders in the same $n$ samples. The $k$th column of $\Z$ is denoted $z_k\in\R^n$ and corresponds to the data for confounding factor $k$.
\item $\X\in\R^{n\times p}$ is a matrix of values for $p$ latent variables to be determined in the same $n$ samples. The $j$th column of $\X$ is denoted $x_j\in\R^n$.
\end{itemize}

To identify the hidden correlation structure of the expression data, we assume a linear relationship between expression levels and the known and latent variables, with random noise added:
\begin{equation}\label{eq:1}
  y_i = \Z v_i + \X w_i + \epsilon_i,
\end{equation}
where $v_i\in \R^d$ and $w_i\in\R^p$ are jointly normally distributed random vectors,
\begin{align}
  p\left(
  \begin{bmatrix}
    v_i\\ w_i
  \end{bmatrix}
  \right) &= \N\left(0,
            \begin{bmatrix}
              \B & \D\\
              \D^T & \A
            \end{bmatrix}\right) \label{eq:7}
\end{align}
with $\B\in\R^{d\times d}$, $\D\in\R^{d\times p}$ and $\A=\diag(\alpha_1^2,\dots,\alpha_p^2)$, such that
\begin{align*}
  \Psi = \begin{bmatrix}
    \B & \D\\
    \D^T & \A
  \end{bmatrix} \end{align*} is a positive semi-definite matrix; the errors $\epsilon_i\in\R^n$ are assumed to be independent and normally distributed, \begin{align*}
  p(\epsilon_i) = \N(0,\sigma^2\I).
\end{align*}
Note that our aim is to identify variance components shared across
genes, and hence $\sigma^2$ is assumed to be the same for all $i$. By
assumption, the errors are also independent of the effect sizes, and
hence we can write
\begin{equation}\label{eq:40}
  p\left(\begin{bmatrix}
    v_i\\ w_i\\ \epsilon_i
  \end{bmatrix}
  \right) = \N\left(0,
            \begin{bmatrix}
              \B & \D & 0\\
              \D^T & \A & 0\\
              0 & 0 & \sigma^2\I
            \end{bmatrix}\right).
\end{equation}
By Lemma~\ref{lem:affine-gaussian}, $y_i$ is normally distributed with
distribution
\begin{equation}\label{eq:39}
  p(y_i) = \N(0,\K) = \frac1{(2\pi)^{\frac{n}2}\sqrt{\det(\K)}} \exp\bigl(-\frac12\langle y_i,\K^{-1}y_i\rangle\bigr),
\end{equation}
where
\begin{align*}
  \K &= 
       \begin{bmatrix}
         \Z & \X & \I
       \end{bmatrix}
       \begin{bmatrix}
         \B & \D & 0\\
         \D^T & \A & 0\\
         0 & 0 & \sigma^2\I
       \end{bmatrix}
       \begin{bmatrix}
         \Z^T \\ \X^T \\ \I
       \end{bmatrix}
  = \Z\B\Z^T + \Z\D\X^T + \X\D^T\Z + \X\A\X^T + \sigma^2 \I,
\end{align*}
and we used the notation $\langle u,v\rangle=u^Tv$ to denote the inner
product between two vectors in $\R^n$.

Defining matrices $\V\in\R^{d\times m}$ and $\W\in\R^{p\times m}$, whose columns are the random effect vectors $v_i$ and $w_i$, respectively, eq.~(\ref{eq:1}) can be written in matrix notation as
\begin{align*}
  \Y = \Z \V + \X \W + \ebf
\end{align*}
Under the assumption that the columns $y_i$ of $\Y$ are independent samples of the distribution (\ref{eq:39}), the likelihood of observing $\Y$ given covariate data $\Z$, (unknown) latent variable data $\X$ and values for the hyper-parameters $\Theta=\{\sigma^2,\A, \B, \D\}$, is given by
\begin{align*}
  p(\Y\mid \Z,\X,\Theta) 
  &= \prod_{i=1}^m p(y_i\mid 0, \K).
\end{align*}
Note that in standard mixed-model calculations, the distribution (\ref{eq:39}) is often arrived at by integrating out the random effects. This is equivalent to application of Lemma~\ref{lem:affine-gaussian}.

 To conclude, the log-likelihood is, upto an additive constant, and divided by half the number of genes:
\begin{equation*}  
  \L =  -\frac2{m}\Bigl[ \frac{m}2\log\det(\K) + \frac12 \sum_{i=1}^m \langle y_i,\K^{-1}y_i\rangle \Bigr] = - \log\det(\K) - \tr\bigl(\K^{-1} \C),
\end{equation*}
where
\begin{align*}
  \C = \frac{\Y\Y^T}{m}
\end{align*}
is the empirical covariance matrix.

\section{Systematic effects on the mean}
\label{sec:syst-effects-mean}

Eq.~\eqref{eq:1} only considers random effects, which leads to a model for studying systematic effects on the covariance between samples. We could also include fixed effects to model systematic effects on mean expression level. However, by centering the data, $\sum_{i=1}^m y_i=0$, the maximum-likelihood estimate of such fixed effects is always zero. To see this, let $\T\in\R^{n\times c}$ be a matrix of $c$ covariates with fixed effects $\beta\in\R^c$ shared across genes (we are only interested in discovering systematic biases in the data). Then the minus log-likelihood \eqref{eq:2} becomes
\begin{align*}
  \L =  \log\det(\K) + \frac1{m} \sum_{i=1}^m \langle y_i-\T\beta,\K^{-1}(y_i-\T\beta)\rangle
\end{align*}
Optimizing with respect to $\beta$ leads to the equation
\begin{align*}
 \hat\beta = ( \T^T \K^{-1} \T)^{-1}\T^T\bar y 
\end{align*}
where
\begin{align*}
  \bar y = \frac1{m}\sum_{i=1}^m y_i=0.
\end{align*}

\section{Solution of the model without latent variables}
\label{sec:solut-model-no-latent}

We start by considering the problem of finding the maximum-likelihood solution in the absence of any latent variables, i.e. minimizing eq.~\eqref{eq:2} with
\begin{equation}\label{eq:10}
  \K = \Z\B\Z^T + \sigma^2\I
\end{equation}
with respect to $\B$ and $\sigma^2$.

Note first of all that we may assume the set of confounding factors $\{z_1,\dots, z_d\}$ to be linearly independent, because if not, the expression in eq.~\eqref{eq:1} can be rearranged in terms of a linearly independent subset of factors whose coefficients are still normally distributed due to elementary properties of the multivariate normal distribution, see for instance the proof of Lemma~\ref{lem:ortho} below. Linear independence of $\{z_1,\dots, z_d\}$ implies that we must have $d\leq n$ and $\rnk(\Z)=d$.

 The singular value decomposition allows to decompose $\Z$ as $\Z = \U\Gbf\V^T $, where $\U\in\R^{n\times n}$, $\U^T\U=\U\U^T=\I$, $\Gbf\in\R^{n\times d}$ diagonal with $\gamma_k^2\equiv\Gamma_{kk}>0$ for $k\in\{1,\dots,d\}$ [this uses $\rnk(\Z)=d$], and $\V\in\R^{d\times d}$, $\V^T\V=\V\V^T=\I$. There is also a `thin' SVD, $\Z= \U_1\Gbf_1\V^T$, where $\U_1\in\R^{n\times d}$, $\U_1^T\U_1=\I$, $\Gbf_1\in\R^{d\times d}$ diagonal with diagonal elements $\gamma_k^2$. In block matrix notation, $\U=(\U_1, \U_2)$ and
  \begin{equation}\label{eq:14}
    \Z =
    \begin{pmatrix}
      \U_1& \U_2      
    \end{pmatrix}
    \begin{pmatrix}
      \Gbf_1\\
      0 
    \end{pmatrix}
    \V
  \end{equation}
  Note that unitarity of $\U$ implies $\U_1^T\U_2=0$. 

 Denote by $\Hs$ the space spanned by the columns (i.e. covariate vectors) of $\Z$. The projection matrix $\Pb$ onto $\Hs$ is given by
  \begin{align*}
    \Pb = \Z(\Z^T\Z)^{-1}\Z^T &= \U_1\Gbf_1 \V^T (\V\Gbf_1^{-2}\V^T) \V\Gbf_1 \U_1^T = \U_1\U_1^T.
  \end{align*}

Using the basis of column vectors of $\U$, we can write any matrix $\M\in\R^{n\times n}$ as a partitioned matrix
\begin{align}
  \U^T\M\U &=
  \begin{pmatrix}
    \M_{11} & \M_{12}\\
    \M_{21} & \M_{22}
  \end{pmatrix} \label{eq:31}\\
  \intertext{where}
  \M_{ij} &= \U_i^T \M \U_j. \label{eq:31b}
\end{align}
The following results for partitioned matrices are derived easily or can be found in \cite{horn1985}:
\begin{align}
  \tr(\M) &= \tr(\M_{11}) + \tr(\M_{22})\label{eq:41}\\
  \det(\M) &= \det\bigl(\M_{11}-\M_{12}\M_{22}^{-1}\M_{21}\bigr) \, \det(\M_{22})\label{eq:42}
\end{align}

Using this notation, the following result solves the model without latent variables:
\begin{thm}\label{thm:no-latent}
  Let $\C\in\R^{n\times n}$ be a positive definite matrix such that
\begin{equation}\label{eq:23}
    \lambda_{\min} (\C_{11}) > \frac{\tr(\C_{22})}{n-d},
  \end{equation}
  where $\lambda_{\min}(\cdot)$ denotes the smallest eigenvalue of a matrix. Then the maximum-likelihood solution
  \begin{align}\label{eq:43}
    \hat\K = \argmin_{\{\K\colon \K = \Z\B\Z^T +\sigma^2\I\}}  \log\det \K + \tr\bigl(\K^{-1}\C\bigr),
  \end{align}
  subject to $\B$ being positive semi-definite and $\sigma^2\geq 0$,  is given by
  \begin{align}
    \hat \B &= \V \Gbf_1^{-1}  (\C_{11}-\hat\sigma^2\I)  \Gbf_1^{-1} \V^T \label{eq:21}\\
    \hat\sigma^2 &= \frac{\tr(\C_{22})}{n-d}\label{eq:22}
  \end{align}
\end{thm}
\begin{proof}
  Using eq.~\eqref{eq:14}, we can write
  \begin{align*}
    \K = \Z\B\Z^T + \sigma^2\I &= \U_1\Gbf_1\V^T\B\V\Gbf_1\U_1^T + \sigma^2 (\U_1\U_1^T + \U_2\U_2^T)\\
    &=  \U_1\Gbf_1\V^T \bigl( \B + \sigma^2 \V\Gbf_1^{-2}\V^T\bigr) \V\Gbf_1\U_1^T + \sigma^2  \U_2\U_2^T.
  \end{align*}
  Hence, in the block matrix notation (\ref{eq:31}), we have
  \begin{align*}
    \K_{11} &= \Gbf_1\V^T \bigl( \B + \sigma^2 \V\Gbf_1^{-2}\V^T\bigr) \V\Gbf_1\\
    \K_{22} &= \sigma^2 \I\\
    \K_{12} &= \K_{21} = 0.
  \end{align*}
  It follows that
  \begin{align*}
    \K^{-1} =
    \begin{pmatrix}
      \K_{11}^{-1} & 0\\
      0 & \K_{22}^{-1}
    \end{pmatrix}
  \end{align*}
  and, using eqs.~(\ref{eq:41}) and (\ref{eq:42}),
  \begin{align*}
    \log\det(\K) &= \log\det(\K_{11}) + \log\det(\K_{22})
                   = \log\det(\K_{11}) + (n-d) \log(\sigma^2)\\
    \tr(\K^{-1}\C) &= \tr(\K_{11}^{-1}\C_{11}) + \tr(\K_{22}^{-1}\C_{22})
                     = \tr(\K_{11}^{-1}\C_{11}) + \frac{\tr(\C_{22})}{\sigma^2}.
  \end{align*}
  Let $\C_{11}$ have eigenvalues $\lambda_1\geq \dots\geq \lambda_d$ with corresponding eigenvectors $u_1,\dots, u_d\in\R^d$. Applying Lemma~\ref{lem:von-Neumann} to the term  $\tr(\K_{11}^{-1}\C_{11})$, it follows that for the minimizer $\hat\K$,  $\hat\K_{11}$ must have eigenvalues $\kappa_1\geq \dots\geq \kappa_d$ with the same eigevectors $u_1,\dots, u_d$ as $\C_{11}$. Expressing the minus log-likelihood in terms of these eigenvalues results in
  \begin{align*}
    \L(\hat\K) = \sum_{i=1}^d\log(\kappa_i) + \sum_{i=1}^d  \kappa_i^{-1} \lambda_i + (n-d) \log(\sigma^2) + \frac{\tr(\C_{22})}{\sigma^2}.
  \end{align*}
  Minimizing with respect to the parameters $\kappa_i$ and $\sigma^2$ (i.e., 
setting their derivatives to zero) results in the solution $\hat\kappa_i=\lambda_i$ for all $i$ and $\hat\sigma^2=\frac{\tr(\C_{22})}{n-d}$. In other words, $\hat\K_{11}$ has the same eigenvalues and eigenvectors as $\C_{11}$, that is,
\begin{align*}
  \hat\K_{11} &= \C_{11}.
\end{align*}
This equation is satisfied if
\begin{align*}
  \hat\B+\hat\sigma^2 \V\Gbf_1^{-2} \V^T &= \V \Gbf_1^{-1} \C_{11} \Gbf_1^{-1} \V^T
 \end{align*}
 or
\begin{align*}
  \hat\B&= \V \Gbf_1^{-1} (\C_{11}-\hat\sigma^2\I)  \Gbf_1^{-1} \V^T
\end{align*}

 $\hat \B$ is positive semi-definite if and only if for all $v\in\R^d$
  \begin{align*}
    0< \langle v, \hat \B v\rangle = \langle w, (\C_{11}-\hat\sigma^2\I) w\rangle,
  \end{align*}
  where $w=\Gbf_1\V v$. Because $\V$ is unitary and $\Gbf_1$ diagonal with strictly positive elements, $\langle v, \hat \B v\rangle>0$ for all $v\in\R^d$ if and only if $\langle w, (\C_{11}-\hat\sigma^2\I) w\rangle>0$ for all $w\in\R^d$, or
  \begin{align*}
    0< \min_{w\in\R^d} \frac{\langle w, \C_{11}w\rangle}{\langle w,w\rangle} -\hat\sigma^2 = \lambda_{\min} (\C_{11}) - \hat\sigma^2.
  \end{align*}
\end{proof}

Eq.~\eqref{eq:23} is a condition on the amount of variation in $\Y$ explained by the confounders $\Z$, with $\lambda_{\min}(\C_{11})$ being (proportional to) the minimum amount of variation explained by any of the dimensions spanned by the columns of $\Z$, and $\frac1{n-d}\tr(\C_{22})$ being the average amount of variation explained by the dimensions orthogonal to  the columns of $\Z$. Failure of this condition simply means that there must be other, latent variables that explain more variation than the known ones, which is precisely what we are seeking to detect.

A useful special case of Theorem~\ref{thm:no-latent} occurs when the number of confounders equals one. In this case, we are seeking maximum-likelihood solutions for $\K$ of the form
\begin{align*}
  \K = \beta^2 z z^T + \sigma^2\I,
\end{align*}
where $z\in\R^n$ is the confounding data vector.  Let $\gamma^2=\|z\|^2$ and $u=\frac1\gamma z$. Then $\P_z=uu^T$ is the projection matrix onto $z$, $\C_{11}=\langle u,\C u\rangle$, and $\tr(\C_{22})=\tr((1-\P_z)\C)=\tr(\C)-\langle u,\C u\rangle$. By Theorem~\ref{thm:no-latent}, we have
\begin{align}
  \hat\beta^2 &= \frac1{\gamma^2} \biggl\{\langle u,\C u\rangle - \frac{\tr\bigl([1-\P_z]\C\bigr)}{n-1} \biggr\}\nonumber\\
  &=  \frac1{\gamma^2} \biggl\{\frac{n}{n-1} \langle u,\C u\rangle - \frac{\tr(\C)}{n-1} \biggr\}\label{eq:34}\\
  \hat\sigma^2 &= \frac{\tr\bigl([1-\P_z]\C\bigr)}{n-1} = \frac{\tr(\C)-\langle u,\C u\rangle}{n-1}, \nonumber
\end{align}
provided
\begin{align*}
  \langle u,\C u\rangle > \frac{\tr(\C)}n.
\end{align*}

\section{Solution of the model without known covariates}
\label{sec:solut-model-no-known}

Next, consider a model without known covariates, i.e.\ with posterior sample covariance matrix $\K=\K_X(\{\alpha_j, x_j\})+\sigma^2\I$, where
\begin{align*}
  \K_X\bigl(\{\alpha_j, x_j\}\bigr) &= \sum_{j=1}^p \alpha^2_j x_j x_j^T.
\end{align*}
This model is equivalent to probabilistic principal component analysis \cite{tipping1999probabilistic,lawrence2005probabilistic}, and its maximum-likelihood solution is given by the first $p$ eigenvectors or principal components with largest eigenvalues of $\C$. Here we present a more direct proof of this fact than what can be found in the literature.

\begin{lem}\label{lem:symm}
  Without loss of generality, we may assume that the latent variables have unit norm, are linearly independent, and are mutually orthogonal.
\end{lem}
\begin{proof}
  If the latent variables do not have unit norm, define $c_j=\|x_j\|^{-1}$, $\alpha'_j=\alpha_j/c_j$ and $x'_j=c_jx_j$ for all $j$. It follows immediately that $\|x'_j\|=1$ and
  \begin{align*}
    \K_X\bigl(\{\alpha_j, x_j\}\bigr) = \K_X\bigl(\{\alpha'_j, x'_j\}\bigr).
  \end{align*}
  Next assume that the  latent variables are not linearly independent, i.e. that $\rnk(\K_X)=r<p$. Because $\K_X$ is a symmetric matrix, we must have $\K_X = \sum_{l=1}^r t_l t_l^T$ for some set of linearly independent vectors $t_l\in\R^n$. Define $\alpha'_l=\|t_l\|$ and $x'_l=t_l/\|t_l\|$. Then $x'_l$ has unit norm and
  \begin{align*}
    \K_X\bigl(\{\alpha'_l,x'_l\}\bigr) = \K_X\bigl(\{\alpha_j, x_j\}\bigr).
  \end{align*}
  Finally, recall that
  \begin{align*}
    \K_X(\{\alpha_j, x_j\})=\X\A\X^T = (\X\A^{\frac12})  (\X\A^{\frac12})^T,
  \end{align*}
  where $\A=\diag(\alpha_1^2,\dots,\alpha_p^2)$. Because we may now assume that $\rnk(\X)=p$, and because $\alpha_j>0$ for all $j$, the matrix $\X\A^{\frac12}$ has singular value decomposition 
\begin{align*}
    \X\A^{\frac12} = \U\Xi \V^T
  \end{align*}
  with $\U\in\R^{n\times p}$, $\U^T\U=\I$, $\Xi\in\R^{p\times p}$ diagonal with diagonal elements $\Xi_{jj}=\xi_j>0$, and $\V\in\R^{p\times p}$, $\V^T\V=\V\V^T=\I$. Hence
  \begin{align*}
    \K_X\bigl(\{\alpha_j,x_j\}) = \U \Xi^2 \U^T = \sum_{j=1}^p \xi_j^2 u_j u_j^T = \K_X\bigl(\{\xi_j,u_j\}),
  \end{align*}
  with $u_j$ the orthonormal columns of $\U$, $\langle u_j, u_j'\rangle=(\U^T\U)_{jj'}=\delta_{j,j'}$.
\end{proof}

We will also need the following simple result:
\begin{lem}\label{lem:ev}
  Let $\lambda_1\geq\lambda_2\geq\dots\geq\lambda_n>0$ be a decreasing sequence of positive numbers, and let $1\leq p< n$. If there exists $j>p$ such that $\lambda_p>\lambda_j$, then
  \begin{equation}\label{eq:28}
    \lambda_p>\frac{1}{n-p}\sum_{j=p+1}^n \lambda_j.
  \end{equation}
\end{lem}
\begin{proof}
  Eq.~\eqref{eq:28} follows from
  \begin{align*}
    \lambda_p-\frac{1}{n-p}\sum_{j=p+1}^n \lambda_j &= \frac{1}{n-p}\sum_{j=p+1}^n (\lambda_p-\lambda_j)>0,
  \end{align*}
  because each term on the r.h.s. is non-negative, and at least one is strictly positive.
\end{proof}

\begin{thm}\label{thm:latent-only}
  Let $\C\in\R^{n\times n}$ be a positive definite matrix with eigenvalues $\lambda_1\geq\dots\geq \lambda_n$ and corresponding eigenvectors $u_1,\dots, u_n$, and let either $p=n$ or $1\leq p< n$ such that there exists $j>p$ with $\lambda_p>\lambda_j$. Then the maximum-likelihood solution 
  \begin{align*}
    \hat\K = \argmin_{\{\K\colon \K = \X\A\X^T +\sigma^2\I\}}  \log\det \K + \tr\bigl(\K^{-1}\C\bigr),
  \end{align*}
  is given by
  \begin{align*}
    \hat x_j &= u_j\\
    \hat\alpha _j^2 &= \lambda_j-\hat\sigma^2\\
    \hat\sigma^2 &= \frac{1}{n-p}\sum_{j=p+1}^n \lambda_j.
  \end{align*}
\end{thm}
\begin{proof}
  By Lemma~\ref{lem:symm}, we can assume that $\X$ has orthonormal columns, and hence there exist $\V\in\R^{n\times (n-p)}$ such that $\Q=(\X,\V)\in\R^{n\times n}$ is unitary, $\Q^T\Q=\Q\Q^T=\I$. Hence $\K=\X\A\X^T+\sigma^2\I$ has the spectral decomposition
  \begin{align*}
   \K  =
  \begin{pmatrix}
    \X & \V   
  \end{pmatrix}
  \begin{pmatrix}
    \A^2 + \sigma^2\I& 0\\
    0 & \sigma^2\I
  \end{pmatrix}
  \begin{pmatrix}
    \X^T\\
    \V^T
  \end{pmatrix},
  \end{align*}
  and hence
  \begin{equation*}
    \K^{-1} = \sum_{j=1}^p \frac1{\alpha_j^2+\sigma^2} x_j x_j^T + \frac1{\sigma^2} \sum_{l=1}^{n-p} v_l v_l^T,
  \end{equation*}
  where $v_l\in\R^n$ are the columns of $\V$.

  Assume that the $\alpha_j^2$ are ordered, $\alpha_1^2\geq \dots\geq \alpha_p^2$.
  Applying von Neumann's Lemma \ref{lem:von-Neumann} gives
  \begin{align}
    \L &= \log\det(\K) + \tr\bigl(\K^{-1}\C\bigr) \nonumber\\
    &\geq \sum_{j=1}^p \log(\alpha_j^2+\sigma^2) +(n-p)\log(\sigma^2) + \sum_{j=1}^p \frac{\lambda_j}{\alpha_j^2+\sigma^2} + \sum_{j=p+1}^n \frac{\lambda_j}{\sigma^2}, \label{eq:26}
  \end{align}
  with equality if and only if
  \begin{align*}
    x_j &= u_j \text{ for }  j=1,\dots,p\\
    v_l &= u_{p+l} \text{ for }  l=1,\dots, n-p      
  \end{align*}
  Hence, independent of the values for $\alpha_j$, the maximum-likelihood latent variables are the eigenvectors of $\C$ corresponding to the  $p$ largest eigenvalues. Minimizing eq,~\eqref{eq:26} w.r.t. $\alpha_j^2$ and $\sigma^2$ then gives
  \begin{align*}
    \alpha_j^2 &= \lambda_j - \sigma^2\\
    \sigma^2 &= \frac1{n-p} \sum_{j=p+1}^N \lambda_j.
  \end{align*}
  By Lemma~\ref{lem:ev}, $\alpha_j^2>0$ for all $j$.
\end{proof}

Note that plugging the maximum-likelihood values in the likelihood function gives
\begin{equation}\label{eq:27}
  \L_{\min} = \sum_{j=1}^p \log(\lambda_j) + (n-p) \log\Bigl(\frac1{n-p}\sum_{j=p+1}^n \lambda_j\Bigr) + n
\end{equation}
Either $p$ can be set \textit{a priori} small enough such that condition~\eqref{eq:28} is satisfied, or else the value of $p$ with smallest $\L_{\min}$ satisfying this condition can be found easily from eq.~\eqref{eq:27}.

Note also that in the models of \cite{tipping1999probabilistic,lawrence2005probabilistic}, uniform prior variances are assumed ($\alpha_1^2=\dots=\alpha_p^2=1$), such that $\X$ is defined upto an arbitrary rotation, because $\X\X^T=(\X\Rm)(\X\Rm)^T$ for any rotation matrix $\Rm$. In our model, there is no such rotational freedom (if $\A$ is assumed to be diagonal), except if $\C$ has eigenvalues with multiplicities greater than one, when there is some freedom to choose the corresponding eigenvectors.

\section{Solution of the full model}
\label{sec:solution-full-model}

\subsection{Orthogonality of known and hidden confounders}
\label{sec:orth-known-hidd}


\begin{lem}\label{lem:ortho}
  Without loss of generality, we may assume that the latent variables are orthogonal to the known confounders:
  \begin{equation}
    \label{eq:30}
    \X^T\Z = \Z\X^T = 0.
  \end{equation}
\end{lem}
\begin{proof}
  As in Section~\ref{sec:solut-model-no-latent}, let $\Pb$ again be the projection matrix on the space spanned by the known covariates $z_k$ (i.e. the columns of $\Z$). For any choice of latent variables $x_j$, we have
\begin{align*}
  x_j = \Pb x_j + (1-\Pb) x_j = \sum_{k=1}^d m_{kj} z_k + \tilde x_j,
\end{align*}
for some matrix of linear coefficients $\M=(m_{kq})\in\R^{d\times p}$, and with $\langle s_k,\tilde x_j\rangle=0$ for all $k$. Or, in matrix notation
\begin{align*}
  \X = \Z\M + \tilde\X,\quad\text{ with }\quad \tilde\X^T\Z = \Z^T\tilde\X=0
\end{align*}

Plugging this in eq.~\eqref{eq:1}, results in
\begin{equation}\label{eq:29}
  y_i = \Z\tilde v_i +\tilde \X w_i + \epsilon_i
\end{equation}
where $\tilde v_i = v_i + \M w_i $. Hence
\begin{align*}
  \begin{bmatrix}
    \tilde v_i\\ w_i\\ \epsilon_i
  \end{bmatrix}
  =
  \begin{bmatrix}
    \I &\ M & 0\\
    0 & \I & 0\\
    0 & 0 & \I
  \end{bmatrix}
  \begin{bmatrix}
     v_i\\ w_i\\ \epsilon_i
   \end{bmatrix}
\end{align*}
and hence, using Lemma~\ref{lem:affine-gaussian}, it follows that
\begin{align*}
  p\left(\begin{bmatrix}
    \tilde v_i\\ w_i\\ \epsilon_i
  \end{bmatrix}
  \right) &= \N\left(0,
            \begin{bmatrix}
              \B + \M\D^T + \D\M^T + \M\A\M^T & \D+\A\M^T & 0\\
              \D^T+ \M\A & \A & 0\\
              0 & 0 & \sigma^2\I
            \end{bmatrix}\right).
\end{align*}
This  is still of exactly the same form as eq.~\eqref{eq:40}. Hence model~\eqref{eq:29} is identical to model~\eqref{eq:1}, but has hidden covariates orthogonal to the known covariates. 
\end{proof}

Note that we can parameterize the model with hidden variables orthogonal to the known confounders, $\Z^T\X=0$, but only if we allow the covariances of their effects on gene expression, $\cov(v_i,w_i)=\D$, to be non-zero. Equivalently, we can parameterize the model such that the random effects of hidden variables are statistically independent of the effects of the known confounders, $\cov(v_i,w_i)=0$, but only if we allow the hidden variables to overlap with the known confounders, $\Z^T\X\neq 0$. Mathematically, the choice of orthogonal hidden factors will be much more convenient.

Note also that a transformation to orthogonal hidden factors always induces non-zero covariances among the known confounders via the term $\M\A\M^T$. Hence an important difficulty with the model where $\B$ is assumed to be diagonal, as used in \cite{fusi2012joint}, comes from the fact that non-orthogonal hidden variables are needed to model off-diagonal covariances between the known confounders. It is much more intuitive to model these directly by assuming a general covariance matrix. 

\subsection{Restricted maximum-likelihood solution for the latent variables}
\label{sec:restr-maxim-likel}

\begin{lem}\label{lem:symm-2}
  Without loss of generality, we may assume that the latent variables have unit norm, are linearly independent, and are mutually orthogonal.
\end{lem}
\begin{proof}
  The proof is identical to the proof of Lemma~\ref{lem:symm} -- it is
  straightforward to verify that the transformation to orthonormal
  variables also do not change the form of the off-diagonal term
  $\Z\D\X^T$ in the covariance matrix $\K$, but merely lead to a
  reparameterization of the matrix $\D$.
\end{proof}

To solve the full model, we follow an approach similar to the standard restricted maximum-lilelihood method for linear mixed models \cite{patterson1971recovery,gumedze2011parameter}: we write the negative log-likelihood function $\L=\log\det(\K)+\tr(\K^{-1}\C)$ as a sum
\begin{equation}\label{eq:44}
  \L = \L_1 + \L_2,
\end{equation}
where $\L_2$ will be the log-likelihood restricted to the subspace
orthogonal to the known confounders $\Z$. We will estimate the latent
variables $\X$ and their effect covariances $\A$ by maximizing $\L_2$,
and estimate the effect covariances $\B$ and $\D$ involving the known
confounders by maximizing $\L_1$. Solving for the latent variables on a restricted subspace is motivated by the observation that if $y\in\R^n$ is a sample from the model~(\ref{eq:1}), that is, $p(y)=\N(0,\K)$, then
\begin{align*}
  \U_2\U_2^T y = \U_2\U_2^T\Z v + \U_2\U_2^T\X w + \U_2\U_2^T\epsilon = \X w + \epsilon'.
\end{align*}
In other words, restricted to the subspace orthogonal to $\Z$, the general model becomes a probablistic PCA model where all variation in the data is explained by the latent variables. 

To obtain the decomposition (\ref{eq:44}), we partition $y\in\R^n$ as $y=(y_1,y_2)^T$, where $y_1=\U_1^Ty\in\R^d$ and $y_2=\U_2^Ty\in\R^{n-d}$, and write
\begin{align*}
  p(y) &= p(y_1,y_2) = p(y_1\mid y_2) p(y_2),\\
  \intertext{or}
  \log p(y) &= \log p(y_1,y_2) = \log p(y_1\mid y_2) + \log p(y_2).
\end{align*}
Hence
\begin{align*}
  \L = -\frac{2}{m} \sum_{i=1}^m \log p(y_i) = \underbrace{-\frac{2}{m} \sum_{i=1}^m \log p(y_{i1}\mid y_{i2})}_{\L_1} \underbrace{-\frac{2}{m} \sum_{i=1}^m \log p(y_{i2})}_{\L_2}
\end{align*}

Using standard results for the marginal and conditional distributions of a multivariate Gaussian, we have
\begin{align*}
  p(y_2) &= \N(0,\K_{22})\\
  p(y_1\mid y_2) &= \N\bigl(\K_{12}\K_{22}^{-1}y_2, (\K_{11}-\K_{12}\K_{22}^{-1}\K_{21})\bigr),
\end{align*}
where we used the partitioned matrix notation of eq.~(\ref{eq:31}).
In particular,
\begin{align*}
  \L_2 &= \log\det(\K_{22}) + \frac{1}m \sum_{i=1}^m \langle \U_2^T y_i, \K_{22}^{-1} \U_2^T y_i \rangle\\
  &= \log\det(\K_{22}) + \frac{1}m \sum_{i=1}^m \tr\bigl(\K_{22}^{-1} \U_2^T y_i y_i^T\U_2\bigl)\\
  &= \log\det(\K_{22}) + \tr\bigl(\K_{22}^{-1}\C_{22}\bigr).
\end{align*}
Note that $\K_{22}=\U_2^T\X\A\X^T\U_2^T + \sigma^2\I$, and hence $\L_2$ depends only on $\X$, $\A$ and $\sigma^2$. The restricted maximum likelihood solution for the latent variables follows immediately:

\begin{thm}\label{thm:lvreml}
  Let $\hat\X\in\R^{n\times p}$, $\hat A\in\R^{d\times d}$,
  and $\hat\sigma^2$ be the solution of
  \begin{align*}
    \{\hat \X,\hat \A, \hat \sigma^2\} = \argmin_{\X,\A, \sigma^2} \L_2(\X,\A,\sigma^2),
  \end{align*}
  where the minimum is taken over all $\X$ with $\X^T\Z=0$, and all positive semi-definite  diagonal matrices $\hat \A$. If there exists $j>p$ such that $\lambda_p>\lambda_j$, then 
\begin{align}
    \hat\X &=\U_2\W_p\label{eq:6}\\
    \hat\A &= \diag(\lambda_1 - \hat\sigma^2, \dots, \lambda_p - \hat\sigma^2)\label{eq:9}\\     
    \hat\sigma^2 &= \frac1{n-d-p}\sum_{j=p+1}^{n-d} \lambda_j\label{eq:13}
  \end{align}
  where $\lambda_1\geq \lambda_2\geq \dots\geq \lambda_{n-d}$ are the
  sorted eigenvalues of $\C_{22}$ with corresponding eigenvectors
  $w_1,\dots,w_{n-d}\in\R^{n-d}$, and
  $\W_p=(w_1,\dots,w_p)\in\R^{(n-d)\times p}$ is the matrix with the
  first $p$ eigenvectors of $\C_{22}$ as columns.
\end{thm}
\begin{proof}
  Defining $\tilde\X=\U_2^T\X\in\R^{(n-d)\times p}$, we have $\K_{22}=\tilde \X\A\tilde\X^T + \sigma^2\I$, and $\L_2$ becomes precisely the minus log-likelihood of the model without known covariates (Section~\ref{sec:solut-model-no-known}), as a function of the latent variables $\tilde \X$ on the \emph{reduced} $(n-d)$-dimensional space orthogonal to the known confounders $\Z$. Hence by Theorem~\ref{thm:latent-only},
  \begin{align*}
    \hat{\tilde \X} &= \W_p\\
    \hat\A &= \diag(\lambda_1 - \sigma^2, \dots, \lambda_p - \sigma^2),
  \end{align*}
  where $\lambda_1\geq \lambda_2\geq \dots\geq \lambda_{n-d}$ are the sorted eigenvalues of $\C_{22}$ and $\W_p\in\R^{(n-d)\times p}$ is the matrix having the corresponding first $p$ eigenvectors as columns. Note that $\hat\A$ is positive semi-definite by Lemma~\ref{lem:ev} and the assumption that there exists $j>p$ such that $\lambda_p>\lambda_j$. It remains to `pull-back' $\tilde\X$ to the original $n$-dimensional space, using the orthogonality condition \eqref{eq:30}:
  \begin{align*}
    \hat\X = (\U_1\U_1^T + \U_2\U_2^T) \hat\X = \U_2\U_2^T \hat\X = \U_2 \hat{\tilde\X} = \U_2\W_p.
  \end{align*}
  This proves eqs.~\eqref{eq:6} and \eqref{eq:9}.
\end{proof}

\subsection{Solution for the variance parameters given the latent variables}
\label{sec:solut-vari-param}

With $\hat \X$, $\hat\A$ and $\hat\sigma^2$ determined by the minimization of $\L_2$ in Theorem~\ref{thm:lvreml}, $\L_2(\hat\X,\hat\A,\hat\sigma^2)$ is  constant in terms of the parameters $\B$ and $\D$ that remain to be optimized. Hence optimizing $\L_1$ with respect to these parameters is the same as optimizing the total negative log-likelihood $\L(\hat\X,\hat\A,\B,\D,\hat\sigma^2)$ w.r.t. $\B$ and $\D$. We have:

\begin{thm}\label{thm:full}
  Let $\hat\B\in\R^{d\times d}$ and $\hat D\in\R^{d\times (n-d)}$ be the solution of
  \begin{align*}
    \{\hat\B,\hat\D\} = \argmin_{\B,\D} \L_1(\hat\X,\hat\A,\B,\D,\hat\sigma^2) = \argmin_{\B,\D} \L(\hat\X,\hat\A,\B,\D,\hat\sigma^2),
  \end{align*}
  subject to the constraint that $\B$ and $\B-\D\hat\A^{-1}\D^T$ are positive semi-definite. If
  \begin{equation}\label{eq:17}
    \lambda_{\min}(C_{11})>\hat\sigma^2,
  \end{equation}
  then
  \begin{align}
    \hat \B &= \V \Gbf_1^{-1} (\C_{11}-\hat\sigma^2\I)  \Gbf_1^{-1} \V^T\label{eq:5}\\
    \hat\D &= \V \Gbf_1^{-1} \C_{12}\W_p
  \end{align}
  where as before
  \begin{align*}
    \Z =
    \begin{pmatrix}
      \U_1& \U_2      
    \end{pmatrix}
    \begin{pmatrix}
      \Gbf_1\\
      0 
    \end{pmatrix}
    \V^T
  \end{align*}
  is the singular value decomposition of $\Z$, and
  $\W_p=(w_1,\dots,w_p)\in\R^{(n-d)\times p}$ is the matrix with the
  first $p$ eigenvectors of $\C_{22}$ as columns.
\end{thm}
\begin{proof}
  Note that the conditions $\B$ and $\B-\D\hat\A^{-1}\D^T$ positive semi-definite are to ensure that the matrix $\begin{pmatrix}
      \B & \D\\
      \D^T & \hat\A
    \end{pmatrix}$
  is positive semi-definite. Next note that with $\hat\X^T$ known, the covariance matrix $\K$ can be written as
  \begin{align*}
    \K =
    \begin{pmatrix}
      \Z & \hat\X
    \end{pmatrix}
    \begin{pmatrix}
      \B & \D\\
      \D^T & \hat\A
    \end{pmatrix}
    \begin{pmatrix}
      \Z^T \\ \hat\X^T
    \end{pmatrix} + \hat\sigma^2 \I
  \end{align*}
  Hence the total log-likelihood is identical to the model with known covariates $\tilde\Z = \begin{pmatrix}
      \Z & \hat\X
    \end{pmatrix}$ and no latent variables (Section~\ref{sec:solut-model-no-latent}). The \emph{unconstrained} maximizing solution (that is, where $\A$ and $\sigma^2$ are also optimized) for the model with known covariates  $\tilde\Z$ is given by Theorem~\ref{thm:no-latent}.  

    Due to $\hat\X^T\Z = 0$ and the definition of $\hat\X$, the singular value decomposition of $\tilde\Z$ is given by
    \begin{align*}
      \tilde\Z =
      \begin{pmatrix}
        \U_1& \hat \X & \U_3      
      \end{pmatrix}
      \begin{pmatrix}
        \Gbf_1 & 0\\
        0 & \I \\
        0 & 0
      \end{pmatrix}
      \begin{pmatrix}
        \V^T & 0\\
        0 & \I
      \end{pmatrix},
    \end{align*}
    where the columns of $\U_3\in\R^{n\times(n-d-p)}$ span the space
    orthogonal to the columns of $\tilde \Z$. Hence the unconstrained
    solution, can be written as
    (cf.~eqs.~(\ref{eq:21})--(\ref{eq:22}))
    \begin{align*}
      \begin{pmatrix}
        \hat\B & \hat\D\\
        \hat\D^T & \hat\A'
      \end{pmatrix} &=
      \begin{pmatrix}
        \V & 0\\
        0 & \I
      \end{pmatrix}
      \begin{pmatrix}
        \Gbf_1^{-1} & 0\\
        0 & \I 
      \end{pmatrix}
      \begin{pmatrix}
        \U_1^T\\
        \hat \X^T
      \end{pmatrix}
      (\C-\hat\sigma'^2\I)
      \begin{pmatrix}
        \U_1 &
        \hat \X
      \end{pmatrix}
      \begin{pmatrix}
        \Gbf_1^{-1} & 0\\
        0 & \I 
      \end{pmatrix}
      \begin{pmatrix}
        \V^T & 0\\
        0 & \I
      \end{pmatrix} \\
      \hat\sigma'^2 &= \frac{\tr(\U_3^T\C\U_3)}{n-d-p}
    \end{align*}
    First note that $\hat\sigma'^2=\hat\sigma^2$, because we can write
    $\U_3=\U_2\W_{\sim p}$, where
    $\W_{\sim p}\in\R^{(n-d)\times(n-d-p)}$ is the matrix with the
    $n-d-p$ last eigenvectors of $\C_{22}$.

    Working out the block matrix product results in:
    \begin{align*}
      \hat\B &= \V \Gbf_1^{-1}\U_1^T (\C-\hat\sigma^2\I) \U_1 \Gbf_1^{-1}\V^T= \V \Gbf_1^{-1} (\C_{11}-\hat\sigma^2\I)  \Gbf_1^{-1}\V^T\\
      \hat\D &= \V \Gbf_1^{-1}\U_1^T \C \hat\X = \V \Gbf_1^{-1} \U_1^T \C\U_2\W_p = \V \Gbf_1^{-1} \C_{12}\W_p\\
      \hat\A' &= \hat\X^T (\C-\hat\sigma^2\I) \hat\X = \W_p^T\U_2^T (\C-\hat\sigma^2\I)  \U_2\W_p = \W_p^T (\C_{22}-\hat\sigma^2\I)\W_p\\
      &= \diag(\lambda_1-\hat\sigma^2,\dots, \lambda_p-\hat\sigma^2)
    \end{align*}
    Hence, also the estimate $\hat\A' =\hat\A$. Because the unconstrained optimization of $\L$ given $\hat\X$ results in the same estimate for $\A$ and $\sigma^2$ as the intial constrained optimization where these parameters were given, it follows that also the estimates of $\B$ and $\D$ must be the same:
    \begin{align*}
      \{\hat\B,\hat\D\}  = \argmin_{\B,\D} \L(\B,\D \mid \hat\X,\hat\A,\hat\sigma^2) = \argmin_{\B,\D} \min_{\A,\sigma^2} \L(\A,\B,\D,\sigma^2\mid \hat\X) .
    \end{align*}
\end{proof}

\subsection{\lvreml\ maximizes the variance explained}
\label{sec:lvreml-maxim-vari}

It is tempting to ask whether the combined solution from Theorems~\ref{thm:lvreml} and \ref{thm:full} optimizes the \emph{total} likelihood among all possible $p$-dimensional sets of latent variables. To address this problem, let $\X\in\R^{n\times p}$ be an arbitrary matrix of latent variables whose columns are normalized, mutually orthogonal and orthogonal to the columns of $\Z$, $\X^T\X=\I$ and $\X^T\Z=0$. Because $\U_2$ is only defined upto a rotation, we can always choose
\begin{align*}
  \U_2 =
  \begin{pmatrix}
    \X & \Q
  \end{pmatrix}
\end{align*}
with $\Q\in\R^{n\times(n-d-p)}$ satisfying $\Q^T\Q=\I$, $\Q^T\X=0$ and $\Q^T\Z=0$. From the proof of Theorem~\ref{thm:full} we immediately obtain:
\begin{prop}\label{prop:covar-ml-X}
  Let $\A(\X)\in\R^{p\times p}$, $\B(\X)\in\R^{d\times d}$, $\D(\X)\in\R^{d\times (n-d)}$ and $\sigma^2(\X)>0$ be the solution of
  \begin{align*}
    \{\A(\X),\B(\X),\D(\X),\sigma^2(\X)\} = \argmin_{\A,\B,\D,\sigma^2} \L(\A,\B,\D,\sigma^2\mid \X).
  \end{align*}
  Then
  \begin{align*}
     \B(\X) &= \V \Gbf_1^{-1} (\C_{11}-\hat\sigma^2\I)  \Gbf_1^{-1}\V^T\\
     \D(\X) &= \V \Gbf_1^{-1}\U_1^T \C \X\\
     \A(\X) &= \X^T (\C-\hat\sigma^2\I)\X\\
    \sigma^2(\X) &= \frac{\tr(\Q^T\C\Q)}{n-d-p}
  \end{align*}
  \qed
\end{prop}

Plugging these values into the negative log-likelihood function results in a function that depends only on $\X$:
\begin{prop}
  Let $\X\in\R^{n\times p}$ be an arbitrary choice of latent variables
  with associated maximum-likelihood estimates for the covariance
  parameters given by Proposition~\ref{prop:covar-ml-X}. Then, upto an additive constant
  \begin{equation}\label{eq:45}
    \L_\X = \log\det\Bigl(\X^T\bigl[\C - \C\U_1(\U_1^T\C\U_1)\U_1^T\C \bigr]\X\Bigr) + (n-d-p)\log\bigl(\hat\sigma^2(\X)\bigr)
  \end{equation}
\end{prop}
\begin{proof}
  Recall from Theorem~\ref{thm:latent-only} that the
  maximum-likelihood estimate for $\K$ given $\X$ and its associated
  maximum-likelihood parameters estimates is given by
  \begin{align*}
    \hat\K(\X) =
    \begin{pmatrix}
      \U_1^T\C\U_1 & \U_1^T\C\X & 0\\
      \X^T\C\U_1 & \X^T\C\X & 0\\
      0 & 0 & \hat\sigma^2\I  
    \end{pmatrix} 
  \end{align*}
  while the covariance matrix $\C$ can be written as
  \begin{align*}
    \C &=
         \begin{pmatrix}
           \U_1^T\C\U_1 & \U_1^T\C\X & \U_1^T\C\Q\\
           \X^T\C\U_1 & \X^T\C\X & \X^T\C\Q\\
           \Q\C\U_1 & \Q\C\X & \Q^T\C\Q
         \end{pmatrix}
  \end{align*}
  Hence
  \begin{align*}
    \L_\X &= \L\bigl(\hat\K(\X)\bigr) = \log\det\bigl(\hat\K(\X)\bigr) + \tr\bigl(\hat\K(\X)^{-1}\C\bigr)\\
    &= \log\det
      \begin{pmatrix}
        \U_1^T\C\U_1 & \U_1^T\C\X\\
      \X^T\C\U_1 & \X^T\C\X
      \end{pmatrix}
      + (n-d-p)\log(\hat\sigma^2) + (d+p) + \frac{\tr(\Q^T\C\Q)}{\hat\sigma^2}\\
    &= \log\det
      \begin{pmatrix}
        \U_1^T\C\U_1 & \U_1^T\C\X\\
      \X^T\C\U_1 & \X^T\C\X
      \end{pmatrix}
      + (n-d-p)\log(\hat\sigma^2) + (d+p) + (n-d-p)
  \end{align*}
  Using equation~(\ref{eq:42}) for the determinant of a partitioned
  matrix, we have
  \begin{align*}
    \log\det
      \begin{pmatrix}
        \U_1^T\C\U_1 & \U_1^T\C\X\\
      \X^T\C\U_1 & \X^T\C\X
    \end{pmatrix}
    &= \log\det(\U_1^T\C\U_1) + \log\det\bigl(\X^T\C\X - \X^T\C\U_1 (\U_1^T\C\U_1)^{-1}\U_1\C\X\bigr)\\
    &= \log\det(\U_1^T\C\U_1) + \log\det\Bigl(\X^T\bigl[\C - \C\U_1(\U_1^T\C\U_1)^{-1}\U_1^T\C \bigr]\X\Bigr).
  \end{align*}
  Ignoring the constants $\log\det(\U_1^T\C\U_1)$ and $n$ which do not depend on $\X$, we obtain eq.~(\ref{eq:45}).
\end{proof}

Due to the determinant term in eq.~(\ref{eq:45}), it is not clear whether the restricted maximum-likelihood solution $\hat\X$ of Theorem \ref{thm:lvreml} (with its associated maximum-likelihood covariance parameters of Theorem \ref{thm:full}) is the absolute minimizer of $\L_\X$,
\begin{align*}
  \hat\X = \argmin_{\X\in\R^{n\times p},\X^T\X=\I,\X^T\Z=0} \L_\X\quad 
  \text{\large ?}
\end{align*}
However, we do have the following result:
\begin{thm}
  The restricted maximum-likelihood solution $\hat\X$ of Theorem
  \ref{thm:lvreml} is the set of $p$ latent variables that minimizes
  the residual variance among all choices of $p$ latent variables,
  \begin{align*}
    \hat\X = \argmin_{\X\in\R^{n\times p},\X^T\X=\I,\X^T\Z=0} \sigma^2(\X)
  \end{align*}
\end{thm}
\begin{proof}
  By Proposition~\ref{prop:covar-ml-X} and the arguments leading up to it, we can write
  \begin{align*}
    \tr(\C_{22}) = \tr(\X^T\C\X) + \tr(\Q^T\C\Q^T)
    = \tr\left( (\U^T_2\X)^T\C_{22}(\U_2^T\X)\right) + \tr\left( (\U^T_2\Q)^T\C_{22}(\U_2^T\Q)\right),
  \end{align*}
  where as before $\C_{22}=\U_2^T\C\U_2$ is the restriction of $\C$ to the $(n-d)$-dimensional subspace orthogonal to the $d$ known covariates, and the columns of $\U_2^T\X$ and $\U_2^T\Q$ span mutually orthogonal subspaces within this $(n-d)$-dimensional space. Hence $(n-d-p)\sigma^2(\X)=\tr(\Q^T\C\Q^T)$ is the trace of $\C_{22}$ over the residual $(n-d-p)$-dimensional space orthogonal to the latent variables, within the subspace orthogonal to the $d$ known covariates. By the Courant-Fisher min-max theorem for eigenvalues \cite{horn1985}, the $(n-d-p)$-dimensional subspace of $\R^{n-d}$ with \emph{smallest} trace is the subspace spanned by the eigenvectors of $\C_{22}$ corresponding to its $(n-d-p)$ smallest eigenvalues. By Theorem~\ref{thm:lvreml}, this is exactly the subspace obtained by choosing $\X$ equal to the restricted maximum-likelihood solution $\hat\X$.  \end{proof}

\section{Selecting covariates and the latent dimension}
\label{sec:select-covar-latent}

Two practical problems remain: how to choose the latent variable dimension parameter $p$ and which known covariates to include?

To choose $p$, we will use the following result:
\begin{lem}\label{lem:trace}
  \begin{align*}
    \tr(\C) = \tr(\hat\K) = \tr(\Z\hat\B\Z^T) + \tr(\hat\X\hat\A\hat \X^T) + n\hat\sigma^2 
  \end{align*}
\end{lem}
\begin{proof}
  Use Theorem~\ref{thm:full} to compute
  \begin{align*}
    \tr(\Z\hat\B\Z) &= \tr\Bigl( \U_1\Gbf_1\V^T \bigl[ \V\Gbf_1^{-1} (\C_{11}-\hat\sigma^2\I) \Gbf_1^{-1} \V^T \bigr] \V\Gbf_1 \U_1^T\Bigr)\\
    &= \tr\bigl(\U_1\C_{11} \U_1^T \bigr) - \hat\sigma^2 \tr\bigl(\U_1\U_1^T\bigr)\\
    &= \tr(\C_{11}) - d \hat\sigma^2,
  \end{align*}
  where the last step uses the cyclical property of the trace and the fact that $\U_1^T\U_1=\I_d$. Likewise, we have
  \begin{align*}
    \tr\bigl(\hat\X\hat\A\hat\X\bigr) &= \tr\Bigl( \U_2\W_p \diag(\lambda_1,\dots,\lambda_p) \W_p^T \U_2^T \Bigr) - \hat\sigma^2 \tr\bigl(\U_2 \W_p \W_p^T \U_2^T\bigr)\\
    &= \sum_{j=1}^p \lambda_j - p\hat\sigma^2\\
    &= \sum_{j=1}^{n-d} \lambda_j - (n-d) \hat\sigma^2\\
    &= \tr(\C_{22}) - (n-d) \hat\sigma^2.
  \end{align*}
  Hence
  \begin{align*}
    \tr(\hat\K) = \tr(\Z\hat\B\Z) + \tr(\hat\X\hat\A\hat\X) + n \hat\sigma^2 = \tr(\C_{11}) + \tr(\C_{22}) = \tr(\C)
  \end{align*}
\end{proof}
Because $\C=(\Y\Y^T)/m$, the eigenvalues of $\C$ are (proportional to) the squared singular values of the expression data $\Y$. Hence $\tr(\Z\hat\B\Z)/\tr(\C)$ is the proportion of variation in $\Y$ explained by the known covariates, $\tr(\hat\X\hat\A\hat\X)/\tr(\C)$ the proportion of variation explained by the latent variables, and $n\hat\sigma^2/\tr(\C)$
is the residual variance.

Our method for determining the number of latent variables lets the user decide \textit{a priori} the minimum amount of variation $\rho$ in the data that should be explained by the known and latent confounders.  It follows that given $\rho$, a ``target'' value for $\sigma^2$ is
\begin{align*}
  \sigma^2(\rho) = \min\Bigl\{ \frac{(1-\rho) \tr(\C)}{n}, \lambda_{\min}(\C_{11}) \Bigr\},
\end{align*}
where the minimum with $\lambda_{\min}(\C_{11})$ is taken to ensure that of condition \eqref{eq:17} remains valid. Because the eigenvalues $\lambda_1,\dots,\lambda_{n-d}$ are sorted, the function
\begin{align*}
 f( p )= \frac1{n-d-p}\sum_{j=p+1}^{n-d} \lambda_j
\end{align*}
increases with decreasing $p$. Hence given $\rho$, we define $\hat p$ as
\begin{align*}
  \hat p = \min \bigl\{p\colon 0\leq p< n-d, \lambda_p>\lambda_{n-d}, f(p)< \sigma^2(\rho)\bigr\},
\end{align*}
that is, we choose $\hat p$ to be the \emph{smallest} number of latent variables that explain \emph{at least} a proporition of variation $\rho$ of $\Y$, while guaranteeing that the conditions for \emph{all} mathematical results derived in this document are valid.

Note that unless all eigenvalues of $\C_{22}$ are identical, $\hat p$ always exists. Once the desired number of latent variables $\hat p$ is defined, the latent factors $\hat{\X}$, the variance parameters $\hat{\A}$, and the residual variance estimate $\hat\sigma^2$ (which will be the largest possible value less than or equal to the target value $\sigma^2(\rho)$) are determined by Theorem~\ref{thm:lvreml}. Once those are determined, the remaining covariance parameters $\hat{\B}$ and $\hat{\D}$ are determined by Theorem~\ref{thm:full}.

A second practical problem occurs when the rank of $\Z$ exceeds the number of samples, such that any subset of $n$ linearly independent covariates explains \emph{all} of the variation in $\Y$. To select a more relevant subset of covariates, we rapidly screen all candidate covariates using the model with a single known covariate (Section~\ref{sec:solut-model-no-latent}) to compute the variance $\hat\beta^2$ explained by that covariate alone (eq.~\eqref{eq:34}).  We then keep only those covariates for which $\hat\beta^2\geq\theta\tr(\C)$, where $\theta>0$ is the second free parameter of the method, namely the minimum amount of variation explained by a known covariate on its own. The selected covariates are ranked according to their value of $\hat\beta^2$, and a linearly independent subset is generated, starting from the covariates with highest $\hat\beta^2$.

\section{Downstream analyses}
\label{sec:downstream-analyses}

The inferred maximum-likelihood hidden factors $\hat\X$ and sample covariance matrix $\hat \K$ are typically used to create a dataset of residuals corrected for spurious sample correlations, to increase the power for detecting eQTLs, or as data-derived endophenotypes \cite{fusi2012joint,stegle2012using}. We briefly review these tasks and how they compare between \lvreml\ and \panama\ hidden factors.

\subsection{Correcting data for spurious sample correlations}

To remove spurious correlations due to the known and latent variance components from the expression data $\Y\in\R^{n\times m}$ (see Section \ref{sec:model}), the residuals $\hat{\mathbf{y}}_i\in\R^n$ for gene $i$ with original data $\mathbf{y}_i$ (a column of $\Y$) are contructed as 
\begin{align*}
  \hat{\mathbf{y}}_i= \hat \K \left(\sigma_{c,i}^2\hat\K + \sigma_{e,i}^2\I\right)^{-1} \hat{\mathbf{y}}_i
\end{align*}
where the variance parameters $\sigma_{c,i}^2$ and $\sigma_{e,i}^2$ are fit separately for each gene $i$ \cite{fusi2012joint}. Hence two solutions for the latent factors that give rise to the same $\hat \K$ (as observed in Section \ref{sec:panama-hidd-fact} for \lvreml\ and \panama) will result in the same residuals.

\subsection{Adjusting for known and latent covariates in  eQTL association analyses}

Two approaches for mapping eQTLs are commonly used in this context. The first approach tests for an association between SNP $\mathbf{s}_j$ and gene $\mathbf{y}_i$ using a mixed model, where the SNP is treated as a fixed effect, constructing likelihood ratio statistics as
\begin{align*}
  \text{LOD}_{i,j} = \log \frac{\N(\mathbf{y}_i \mid \theta \mathbf{s}_j, \sigma_{c,i}^2\hat\K + \sigma_{e,i}^2\I)}{\N(\mathbf{y}_i \mid 0, \sigma_{c,i}^2\hat\K + \sigma_{e,i}^2\I)},
\end{align*} 
where the variance parameters $\sigma_{c,i}^2$ and $\sigma_{e,i}^2$ are fit separately for each gene $i$ \cite{fusi2012joint}. Hence for latent factor solutions that give rise to the same $\hat \K$ the association analyses will again be identical.

The second approach performs a linear regression of a gene's expression data, typically using the corrected data $\hat{\mathbf{y}}_i$, on the SNP genotypes $\mathbf{s}_j$, using the known and inferred factors as covariates \cite{stegle2012using}, that is, a linear model is fit where
\begin{equation}\label{eq:lin-assoc}
  \hat{\mathbf{y}}_i = \beta_{i,j} \mathbf{s}_j + \Z \mathbf{a}_i + \hat\X \mathbf{b}_i + \epsilon_i
\end{equation}
where $\Z$ and $\hat\X$ are the matrices of known and estimated latent factors, respectively, and $\mathbf{a}_i\in\R^d$ and $\mathbf{b}_i\in\R^p$ are their respective regression coefficients.

Since maximum-likelihood solutions for the hidden factors by \lvreml\ and \panama\ differ by a linear combination with the known factors $\Z$ that transforms models with hidden factors orthogonal to $\Z$ to equivalent models with hidden factors overlapping with $\Z$, and vice versa (see Section \ref{sec:orth-known-hidd}), it is clear that the same linear transformation will also result in equivalent linear association models in eq. \eqref{eq:lin-assoc}. Hence this type of analysis will also be equivalent between the hidden factors inferred by both approaches.

\subsection{Mapping the genetic architecture of latent variables}

Inferred latent variables are sometimes treated as endophenotypes whose genetic architecture is of interest. In this case SNPs are identified that are strongly associated with the latent variables. Different solutions for the latent variables will then clearly result in different sets of significantly associated SNPs.

Using the maximum-likelihood \lvreml\ inferred latent variables that are orthogonal to known confounders is advantageous in this context, because
\begin{itemize}
  \item The \lvreml\ latent variables are \emph{uniquely defined}. All other solutions that give rise to the same covariance matrix estimate $\hat{\K}$ can be written as a linear combination of the known covariates and the \lvreml\ covariates (see Section \ref{sec:orth-known-hidd}).
  \item When interpreting associated SNPs, there is no risk of attributing biological meaning to a latent variable that is due to the signal coming from the overlapping known covariates.
\end{itemize}
To remove the dependence of genetic association analyses on the choice of equivalent sets of latent variables, we recommend performing a multi-trait GWAS on the joint set of known and latent confounders. If the standard multivariate association test based on canonical correlation analysis \cite{ferreira2009multivariate} is used, results will again be identical between equivalent choices of latent variables, because together with the known confounders they all span the same linear subspace.

\end{document}